\documentclass[12pt]{article}
\usepackage{graphicx}
\usepackage{amsmath} 
\usepackage{amsthm}
\usepackage{amsfonts}
\usepackage{amssymb}
\usepackage{graphicx}
\usepackage{color}
\usepackage{boxedminipage}
\usepackage{algorithm}
\usepackage{bbm}
\usepackage[top=0.5in, bottom=0.8in, left=1in, includeheadfoot]{geometry}
\RequirePackage[numbers]{natbib}
\RequirePackage[colorlinks,citecolor=blue,urlcolor=blue]{hyperref}
\RequirePackage{hypernat}

\theoremstyle{plain}
\newtheorem{theorem}{Theorem}
\newtheorem{thm}{Theorem}
\newtheorem{proposition}{Proposition}
\newtheorem{lemma}{Lemma}

\newtheorem{corollary}[thm]{Corollary}

\theoremstyle{definition}
\newtheorem{definition}{Definition}

\newcommand{\ip}[2]{\langle #1,#2 \rangle}

\newcommand{\bydef}{\stackrel{\bigtriangleup}{=}}
\newcommand{\di}{{_\diamond}}

\newcommand{\prob}[1]{\mathbb{P}\left[{#1}\right]}

\begin{document}

\title{Wireless Scheduling with Partial Channel State Information:
  Large Deviations and Optimality \footnote{This work was partially
    supported by NSF grants CNS-0721380, CNS-1017549, CNS-1161868,
    CNS-0721532, and EFRI-0735905. A preliminary version of this work
    appears in the proceedings of the 31st IEEE International
    Conference on Computer Communications (IEEE INFOCOM), Orlando,
    March 2012 \cite{GopCarSha12:largedev}.}}

\author{Aditya Gopalan\thanks{Department of Electrical Engineering,
    Technion - Israel Institute of Technology, Haifa 32000,
    Israel. Email: {{\tt aditya@ee.technion.ac.il}}} \and Constantine
  Caramanis \thanks{Department of Electrical and Computer Engineering,
    The University of Texas at Austin, Austin, Texas 78712,
    USA. Email: {{\tt constantine@utexas.edu}}} \and Sanjay Shakkottai
  \thanks{Department of Electrical and Computer Engineering,The
    University of Texas at Austin, Austin, Texas 78712,
    USA. Email: {{\tt shakkott@austin.utexas.edu}}} }


\maketitle

\begin{abstract}
  We consider a server serving a time-slotted queued system of
  multiple packet-based flows, where not more than one flow can be
  serviced in a single time slot. The flows have exogenous packet
  arrivals and time-varying service rates. At each time, the server
  can observe instantaneous service rates for only a \emph{subset} of
  flows (selected from a fixed collection of \emph{observable}
  subsets) before scheduling a flow in the subset for service. We are
  interested in queue-length aware scheduling to keep the queues
  short. The limited availability of instantaneous service rate
  information requires the scheduler to make a careful choice of which
  subset of service rates to sample. We develop scheduling algorithms
  that use only partial service rate information from subsets of
  channels, and that minimize the likelihood of queue overflow in the
  system. Specifically, we present a new joint subset-sampling and
  scheduling algorithm called \emph{Max-Exp} that uses only the
  current queue lengths to pick a subset of flows, and subsequently
  schedules a flow using the Exponential rule. When the collection of
  observable subsets is disjoint, we show that Max-Exp achieves the
  best exponential decay rate, among all scheduling algorithms that
  base their decision on the current (or any finite past history of)
  system state, of the tail of the longest queue. To accomplish this,
  we employ novel analytical techniques for studying the performance
  of scheduling algorithms using partial state, which may be of
  independent interest. These include new sample-path large deviations
  results for processes obtained by non-random, predictable sampling
  of sequences of independent and identically distributed random
  variables. A consequence of these results is that scheduling with
  partial state information yields a rate function significantly
  different from scheduling with full channel information. In the
  special case when the observable subsets are singleton flows, i.e.,
  when there is effectively no \emph{a priori} channel-state
  information, Max-Exp reduces to simply serving the flow with the
  longest queue; thus, our results show that to always serve the
  longest queue in the absence of any channel-state information is
  large-deviations optimal.  
\end{abstract}

\section{Introduction}
\label{sec:intro}
Next-generation wireless cellular systems such as LTE-Advanced
\cite{parast09:lteimt} and Wi MAX \cite{andrews} promise high-speed
packet-switched data services for a variety of applications, including
file transfer, peer-to-peer sharing and real-time audio/video
streaming. This demands effective scheduling in typical wireless
environments with time-varying channels and limited resources, to
guarantee high data rates to the users. Together with maximizing data
rates or throughput, the scheduling algorithm at the cellular base
station must keep packet delays in the system low, in order to support
highly delay-sensitive applications like real-time video streaming.

There has been much recent work to develop wireless scheduling
algorithms with optimal throughput and/or delay performance
\cite{taseph92,stol:async,stol:ldexp,sadiqgdv:pseudolog,vjvlin:ldlyapunov}. Such
opportunistic scheduling algorithms utilize instantaneous wireless
Channel State Information (CSI) from all users to make good scheduling
decisions.  However, in a practical situation with a large number of
users in the network, channel state feedback resources could
potentially be limited, i.e., it might be infeasible to acquire
complete instantaneous CSI from all channels due to bandwidth and
latency limitations. Instead, it might be possible to request CSI
feedback from only a \emph{subset of users} each time. Thus, it is
important to develop algorithms that can schedule using only partial
CSI rather than complete CSI, and at the same time afford the best
possible delay performance.

Using partial CSI -- from subsets of channels -- entails a new
dimension of opportunism in wireless scheduling. The scheduling
algorithm needs to make a careful choice of which subsets to sample,
together with how to use the sampled CSI for scheduling. Recently,
natural extensions of complete-CSI scheduling algorithms to the
partial-CSI setting have shown to have throughput-optimal properties
\cite{gopcarsha12:subsets}, yet it is not clear how they perform in
the sense of packet delays. The general structure of low-delay,
partial-CSI scheduling algorithms remains unknown, i.e., how an
algorithm should choose ``good'' subsets of channels, whether any
additional backlog or statistical information is needed for picking
subsets, and if so, how much, how users should be scheduled in the
observed subset etc.

In this work, we develop algorithms for wireless scheduling that use
only partial CSI, {\color{black} i.e.,} from subsets of channels, and
that also enjoy high performance guarantees. We consider a wireless
downlink where a base station schedules users using partial CSI from
subsets of channels.  Viewing the system queue lengths as a surrogate
for packet delays, we seek scheduling strategies that can keep the
longest queue in the system as short as possible, i.e., minimize the
likelihood of overflow of the longest queue. We {\color{black} design}
a new scheduling algorithm, that we term \emph{Max-Exp}, that obtains
partial CSI relying on just current queue lengths and no other
auxiliary information.  Employing sample-path large deviations
techniques, we show that when the observable channel subsets are
disjoint, Max-Exp yields the best decay rate for the longest-queue
overflow probability, across all scheduling strategies which use
subset-based CSI to schedule users. To the best of our knowledge, this
is the first work that analyzes queue-overflow performance for
scheduling with the information structure of partial CSI, and that
provides a simple scheduling algorithm needing no extra statistical
information which is actually rate-function optimal for buffer
overflow.

From a technical standpoint, sample-path large deviations techniques
have successfully been used to analyze wireless scheduling algorithms
\cite{berpastsi98:optcontrol,stol:ldexp,sadiqgdv:pseudolog,vjvlin:ldlyapunov};
yet, significant new analytical challenges emerge when studying the
large deviations behavior of scheduling strategies that cannot access
the full state of the system. A chief difference in this regard arises
from the fact that when scheduling is carried out by observing the
\emph{complete} state/randomness of the system, large deviations occur
depending on how the scheduler responds to atypical channel state
behavior. In other words, a natural cause-effect relationship between
the channel state process and scheduling actions is the basis for the
analysis of large deviations performance. On the other hand, when {\em
  partial} channel state is \emph{acquired} selectively by a
scheduling algorithm, this cause-effect sequence is reversed -- it is
the algorithm that first decides what part of the channel state to
sample; subsequently, this dynamic portion of the channel state can
respond by behaving atypically. Viewed differently, the scheduling
information structure no longer falls into an ``experts'' setting (all
channel rates known in advance) but rather into a ``bandit'' setting
(only chosen channel rates known) \cite{ManSha11:banditexpert},
implying a fundamental change in the large deviations
dynamics. Indeed, we are able to show that this difference results in
a significantly different rate function than that encountered in the
former complete-CSI case.

Also, the standard approach of analyzing queue overflow probability
exponents using continuity of queue-length/delays as functions of the
arrivals and channel processes
\cite{shakkottai:effective,yingsridull:ldwless} becomes cumbersome due
to the complex two-stage sampling and scheduling structure of
scheduling with partial CSI. Thus, we are led to develop new
sample-path large deviations results for processes with dynamically
(and predictably) sampled randomness, which help to bound the
resulting rate functions via connections to appropriate variational
problems. We believe that these techniques and results are of
independent interest as tools to analyze the behavior of scheduling
policies that can only sample parts of the system state.


\subsection{Related Work}
For scheduling with complete CSI, there is a rich body of work on
throughput-optimal scheduling algorithms, starting from the pioneering
approach of Tassiulas et al. \cite{taseph92} to develop the
Backpressure algorithm. A host of scheduling algorithms such as
Max-Weight/Backpressure \cite{taseph92,stol:async}, the Exponential
rule \cite{shasto01,shasristo04,stol:ldexp} and the Log rule
\cite{sadiqgdv:pseudolog} have been developed for scheduling using
full CSI. Many optimality results are now known for the
delay/queue-length performance of the above full-CSI algorithms. These
include expected queue length/delay bounds via Lyapunov function
techniques \cite{neemodroh05,Eryilmaz05stablescheduling}, tail
probability decay rates for queue lengths
\cite{stolram:lwdfld,vjvlin:structuralldp,sadiqgdv:pseudolog,shakkottai:effective,vjvlin:ldlyapunov,stol:ldexp,yingsridull:ldwless},
heavy-traffic optimality \cite{stolyar04} etc.

Throughput-maximizing scheduling has been studied with different forms
of partial CSI, including infrequent channel state measurements
\cite{sarkaretal:infrequent}, group/ random-access based quantized
channel state feedback \cite{gdv06,tangheath05}, optimal channel state
probing with costs \cite{chaproasnkar09,chli07}, delayed CSI
\cite{yinsha11:delayed} and subset-based CSI
\cite{gopcarsha12:subsets}. However, to date, neither the structure
nor performance results for queue overflow tails under scheduling with
partial CSI are known.

\subsection{Contributions}
We describe a new scheduling algorithm -- Max-Exp -- for scheduling
over a wireless downlink when Channel State Information (CSI) is
restricted to a collection of {\em observable channel
  subsets}. Max-Exp picks a subset of channels to observe their
states, depending on an appropriate exponentiated sum of the subset
queue lengths. Having done that, it uses the well-known Exponential
rule \cite{shasto01} to schedule a user from the subset using the
obtained CSI. Thus, Max-Exp \emph{does not} need any additional
information (e.g. traffic/channel statistics) other than queue lengths
to dynamically pick subsets, and only the instantaneous subset channel
states to schedule users.

Our main contributions can be summarized as follows:
\begin{enumerate}
\item We derive a lower bound on the rate function for overflow of the
  longest queue under the Max-Exp scheduling algorithm, using
  sample-path large deviations tools and their connection to
  variational optimal-control problems. A key technical contribution
  here is developing large deviations properties for processes
  obtained by predictably sampling independent and identically
  distributed (iid) sequences. These results help to show that the
  sample-path large deviations rate function, for algorithms that
  sample portions of the channel state, not only depends on the
  standard Cram\'{e}r empirical rate functions of the sampled
  portions, but also relies crucially on the \emph{sampling
    frequencies} of the portions.

  Conversely, we also show \emph{universal} (i.e., over \emph{all}
  scheduling algorithms that use partial, subset-based CSI) upper
  bounds on the rate function of queue overflow. Here again, a
  technical challenge arises due to the fact that for an
  arbitrary\footnote{In the context of this work, an arbitrary
    scheduling algorithm is to be understood as any map that is based
    on the current (or any finite past history) of system state.}
  scheduling algorithm, the large-deviations ``cost'' of buffer
  overflow depends crucially on its subset \emph{sampling behavior} --
  different scheduling algorithms could sample subsets with vastly
  differing frequencies resulting in potentially different costs to
  twist channel state distributions of subsets, and hence different
  rate functions. We develop a novel martingale-based technique to
  quantify this effect and derive a universal upper bound on the
  buffer overflow exponent.

\item In the case where the collection of observable subsets available
  to the scheduler is disjoint, we prove that the lower bound on the
  large deviations buffer overflow rate function for Max-Exp matches
  the uniform upper bound on the rate function over all algorithms.
  This not only characterizes the exact buffer overflow exponent of
  the Max-Exp algorithm, but also shows rather surprisingly that
  \emph{the simple Max-Exp strategy yields the optimal overflow
    exponent across all scheduling rules using partial
    CSI}\footnote{By optimal, we mean optimal among all scheduling
    algorithms that base their decision on the current (or any finite
    past history of) system state.}. As a side consequence, this shows
  that for scheduling with singleton subsets of users, merely
  scheduling the user with the longest queue at each time slot -- a
  greedy strategy when no CSI is available beforehand -- is
  large-deviations rate function-optimal.

  Technically, showing that the lower and upper bounds for the queue
  overflow rate function match involves solving a complex and
  non-convex variational problem arising from the rate function for
  predictably sampled random processes, and is another contribution of
  this work.
\end{enumerate}

  
\section{System Model}
\label{sec:notation}
This section describes the wireless system model we use along with its
associated statistical assumptions. We consider a standard model of a
wireless downlink system \cite{stol:async}: a time-slotted system of
$N$ users serviced by a single base station or server across $N$
communication channels. In each time slot $k \in \{0,1,2,\ldots\}$,
the dynamics of the system are governed by three primary components:
\begin{enumerate}
\item {\bf Arrivals:} An integer number of data packets $A_i(k)$
  arrives to user $i$, $i = 1,\ldots,N$. Packets get queued at their
  respective users if they are not immediately transmitted.
\item {\bf Channel states:} The set of $N$ channels assumes a random
  {\em channel state} $R(k)$, i.e. an $N$-tuple of integer {\em
    instantaneous service rates}. At time slot $k$, we denote the
  instantaneous service rates by $(R_1(k),\ldots,R_N(k))$.
\item {\bf Scheduling:} One user $U(k) \in \{1,\ldots,N\}$ is picked
  for service, and a number of packets not exceeding its instantaneous
  service rate is removed from its queue. Let $D_i(k)$ denote whether
  user $i$ is scheduled in time slot $k$ ($D_i(k) = 1$), or not
  ($D_i(k) = 0$). Then, user $i$'s queue length (denoted by
  $Q_i(\cdot)$) evolves as $Q_i(k+1) = [Q_i(k) + A_i(k) -
  D_i(k)R_i(k)]^+$, where $x^+ \equiv \max(x,0)$.
\end{enumerate}
We assume the following about the stochastics of the arrival
and channel state processes: \\
{\bf Assumption 1 (Arrivals):} Each user $i$'s arrival process
$(A_i(k))_{k=0}^\infty$ is deterministic and equal to $\lambda_i$ at
all time slots. This is done merely for notational simplicity --
{\color{black} any bounded, iid arrival process $(A_i(k))_{k=0}^\infty$
  works, with the only modification being the large-deviations rate
  function of $A_i(k)$ added to all the rate function expressions in
  the
  paper.} \\
{\bf Assumption 2 (Channel States):} The joint channel states $R(k)$,
$k = 0,1,2,\ldots$ are independent and identically distributed across
time, and take values from a finite set $\mathcal{R}$ of integer
$N$-tuples.  Note that the channel states can have \emph{any} joint
distribution and can thus be correlated \emph{across channels/users}. 

\noindent {\bf Scheduling Model:} Under scheduling with {\em partial
  channel state information}, a scheduling {\em algorithm} is defined
to be a rule that, at each time slot $k$, makes \emph{two} sequential
choices to schedule a user:
  \begin{itemize}
  \item {\em Step 1:} Pick a {\em subset} $S(k)$ of the $N$ channels,
    from a given collection $\mathcal{O}$ of {\em observable
      subsets}\footnote{The collection of observable subsets models
      the collection of subsets of channels for which the wireless
      scheduler can obtain instantaneous channel state information, as
      described in the Introduction (\ref{sec:intro}).}. This choice
    can depend on all random variables in time slots up to and
    including $k$ except the channel state $R(k)$.
    \item {\em Step 2:} Once the subset $S(k)$ of channels is chosen,
      the instantaneous service rates $(R_i(k))_{i \in S(k)}$ are
      revealed/available to the scheduling algorithm, and it chooses a
      user $U(k) \in S(k)$ for service, possibly depending on these
      service rates.
  \end{itemize}
  Note that at each time, the channel state information available to
  the scheduling algorithm is restricted to the chosen subset $S(k)$
  of channels, as opposed to the full CSI case where all the service
  rates $(R_i(k))_{i=1}^N$ are available. For further detailed
  discussion about this scheduling model and how it abstracts limited
  channel state information at the wireless physical layer, etc., we
  refer the reader to \cite{gopcarsha12:subsets}.

\section{Objective, Algorithms and Main Results}
\label{sec:mainresults}
Our focus is to design scheduling algorithms that reduce the
likelihood of large queues in the system. Specifically, we seek to
minimize the stationary probability (when it exists) that the longest
queue in the system $||Q(k)||_\infty \bydef \max_{i} Q_i(k)$ exceeds a
threshold $n$. Alternatively, our goal is to maximize the
\emph{exponent} or \emph{decay rate} of the exceedance probability
\[ I \bydef - \lim_{n \to \infty} \frac{1}{n} \log
\mathbb{P}\left[||Q(k)||_\infty \geq n \right]\] (when the limit
exists), for scheduling algorithms that {\em observe only partial
  channel state} while scheduling. Note that for large $n$,
$\mathbb{P}[||Q(k)||_\infty \geq n] \approx e^{-n I}$, so maximizing
the exponent $I$ gives smaller overflow probabilities. Also, it is
well-known that packet delays are closely related to queue lengths
\cite{vjvlin:ldlyapunov}, which justifies using $I$ as our performance
objective.


With this objective in mind, we introduce a new scheduling algorithm
\emph{Max-Exp} (Algorithm \ref{alg:maxexp}). The algorithm may be
interpreted as locally (in Step 2) using the Exponential scheduling
rule \cite{shasristo02,shakkottai:exprule,shasristo04,shasto01}, and
globally (in Step 1) using the Exponential rule metric without the
(observed) instantaneous rate to pick a subset of channels. 

It is well-known that in case the entire set of channels is observable
(i.e., the full-information setting), the Exponential rule maximizes
the exponent of the queue overflow probability \cite{stol:ldexp},
hence it is a natural candidate for the in-subset scheduling rule in
Step 2 of the Max-Exp algorithm. The rule used to choose subsets in
Step 1 is chosen so as to match the in-subset Exponential rule, and
guarantees properties that are required in the fluid limit scaled
description of the dynamics in order to show our main result.

\begin{algorithm}[htbp]
  \caption{Max-Exp}
  \label{alg:maxexp}
  At each time slot $k$, breaking ties arbitrarily,
  \begin{enumerate}
  \item Choose a subset $S(k)$, from the collection $\mathcal{O}$ of
    observable subsets, such that
    $$\sum_{i \in S(k)} \exp \left(\frac{Q_i(k)}{1 +
      \sqrt{\overline{Q}(k)}}\right)$$ is maximized (here
    $\overline{Q}(k) \bydef \frac{1}{N} \sum_{i=1}^N Q_i(k)$ is the
    length of the average queue at time slot $k$).

  \item Schedule a user $i \in S(k)$ such that $R_i(k) \exp
    \left(\frac{Q_i(k)}{1 + \sqrt{\overline{Q}(k)}}\right)$ is
    maximized (the Exponential rule \cite{shasto01}).
  \end{enumerate}
\end{algorithm}

By our probabilistic assumptions on the channel state process, Max-Exp
makes the vector process of queue lengths at each time a discrete-time
Markov chain. Following standard convention
\cite{stol:async,neemodroh05,Eryilmaz05stablescheduling}, we term the
set of arrival rates $\lambda \equiv (\lambda_i)_{i=1}^N$ for which
this Markov chain is positive-recurrent as the \emph{throughput
  region} of Max-Exp. To not deviate from the main focus of this work,
we state that {\color{black} when the observable subsets are
  disjoint,} the throughput region of Max-Exp contains that of any
other scheduling algorithm, i.e., Max-Exp is
\emph{throughput-optimal}\footnote{We mean throughput-optimal among
  all scheduling algorithms that base their decision on the current
  (or any finite past history of) system state.}. {\color{black} The
  proof of throughput-optimality is analogous to that of the
  Max-Sum-Queue scheduling algorithm \cite{gopcarsha12:subsets}.} Our
main result states that Max-Exp yields the best (exponential) rate of
decay of the tail of the longest queue over all strategies that use
partial CSI from disjoint subsets:

\begin{theorem}[{\color{black} Large Deviations Optimality of Max-Exp}]
\label{thm:overallmaxexp}
Let the system's arrival rates $\lambda$ lie in the interior of the
throughput region of the Max-Exp scheduling algorithm. There
exists $J_* > 0$ such that the following holds.
\begin{enumerate}
\item Let $\mathbb{P}$ denote the stationary probability distribution
  that the Max-Exp algorithm induces on the vector of queue lengths.
  Then,
  \[ -\limsup_{n \to \infty} \frac{1}{n} \log \mathbb{P}\left[||Q(0)||_\infty \geq n \right] \geq J_*.\]
\item Let $\pi$ be an arbitrary scheduling rule\footnote{An arbitrary
    scheduling rule is any map that is based on the current (or any
    finite past history) of system state.} that induces a stationary
  distribution $\mathbb{P}^\pi$ on the vector of queue lengths. If the
  system of observable subsets $\mathcal{O}$ is disjoint, then
  \[ -\liminf_{n \to \infty} \frac{1}{n} \log
  \mathbb{P}^\pi\left[||Q(0)||_\infty \geq n\right] \leq J_*.\] 
\end{enumerate}
\end{theorem}

Thus, Max-Exp has the optimal large-deviations exponent (equal to
$J_*$) over all stabilizing scheduling policies with subset-based
partial channel state information.

Theorem \ref{thm:overallmaxexp} highlights the striking property that
Max-Exp, using only current queue length information to sample channel
subsets and the Exponential rule to schedule a sampled channel, yields
the \emph{fastest decay} of the buffer overflow probability across the
whole spectrum of partial-CSI scheduling algorithms -- including those
that potentially use additional statistical information, traffic
characteristics etc. The crucial scheduling step in Max-Exp is Step 1,
which essentially samples the ``right'' channel subset depending on
queue lengths. The result shows that queue length feedback is
sufficient to guarantee good delay performance, provided suitable
subsets of channel states are sampled as with the Max-Exp scheduling
algorithm. We remark that the optimality of Max-Exp continues to hold
even when all queue lengths are delayed by any bounded
amount. {\color{black} We also remark that the restriction to disjoint
  observable subsets is necessary since otherwise, even
  throughput-optimality of Max-Exp-style scheduling rules does not
  hold \cite{gopcarsha12:subsets}. This is briefly because the
  geometry of the throughput region is fundamentally different when
  subsets are disjoint, and its properties play a key role in the
  optimality proof for Max-Exp here.}

En route to proving Theorem \ref{thm:overallmaxexp}, we develop lower
bounds for the large deviations exponents of partially and
deterministically sampled iid processes, that are of independent
interest. This results in a new rate function formulation in terms of
variational optimization, that differs significantly from existing
rate functions
\cite{vjvlin:structuralldp,sadiqgdv:pseudolog,shakkottai:effective,vjvlin:ldlyapunov,stol:ldexp,yingsridull:ldwless}
by explicitly incorporating partial channel state sampling
behavior. Standard optimal control approaches for the full-CSI case
cannot be applied to analyze partial-CSI scheduling algorithms --
since only a portion of the channel state is revealed to the
scheduler, the channel state process can cause large deviations by
behaving atypically just in the revealed portion, and not jointly as a
whole.

A related challenge arises in the process of finding universal upper
bounds on the decay rate for arbitrary partial-CSI scheduling
policies\footnote{In the context of this work, an arbitrary scheduling
  policy is to be understood as any map that is based on the current
  (or any finite past history) of system state.}. Recent
large-deviations work in full-CSI scheduling
\cite{vjvlin:ldlyapunov,stol:ldexp} accomplishes this by calculating
the ``cost'' of universal channel-state sample paths that cause buffer
overflow under any scheduling algorithm; however, this procedure fails
for algorithms actively sampling the channel state, since the cost of
such sample paths intimately depends on the subset sampling
behavior. To overcome this, we use a martingale-based argument in a
novel way with the standard exponential tilting method to prove
universal upper bounds on the exponent.

Observe that Max-Exp reduces to the following Max-Queue scheduling
algorithm when the observable subsets are all the singleton users:
\begin{algorithm}[htbp]
  \caption{Max-Queue}
  \label{alg:maxq}
  At each time slot $k$, breaking ties arbitrarily,
  \begin{enumerate}
  \item Schedule a user $i$ such that $Q_i(k)$ is maximized.
  \end{enumerate}
\end{algorithm}

Thus, an immediate corollary of Theorem \ref{thm:overallmaxexp} is the
following optimality result for Max-Queue when the observable user
subsets are restricted to singletons, i.e., when there is effectively
no CSI to use in scheduling:

\begin{corollary}[Large Deviations Optimality of Max-Queue for
  singleton observable subsets]
  \label{cor:overallmaxqueue}
  If the system's arrival rates $\lambda$ lie in the interior of the
  throughput region of the Max-Queue scheduling algorithm, then
  Max-Queue has the optimal\footnote{By optimal, we mean optimal among
    all scheduling algorithms that base their decision on the current
    (or any finite past history of) system state.} large-deviations
  exponent of the queue overflow probability over all stabilizing
  scheduling policies that can sample only individual channel states.
\end{corollary}

{\em Road map to prove Theorem \ref{thm:overallmaxexp}:} Though
Theorem \ref{thm:overallmaxexp} for Max-Exp is our chief result, we
prove it by first establishing the optimality result for Max-Queue
(Corollary \ref{cor:overallmaxqueue}), and then extending the argument
to the setting of general disjoint subsets. This is mainly because the
essence of the optimality lies in the key subset selection step, and
restricting attention to the case of singleton observable subsets
allows us to concentrate on how subset selection influences the large
deviations rate function of buffer overflow. Technically, another
reason for this order of working is that Max-Queue can naturally be
analyzed with the standard $O(n)$ fluid scaling, whereas showing the
optimality property for Max-Exp requires using a more delicate fluid
limit framework at the $O(\sqrt{n})$ ``local'' fluid time-scale
\cite{shasto01,stol:ldexp}.

\section{Preliminaries and Sample Path Large Deviations Framework}
\label{sec:ldframework}
This section lays down preliminaries for the sample-path large
deviations techniques we use to study overflow probabilities of
wireless scheduling algorithms. Much of this framework is standard in
large deviations analyses of wireless systems
\cite{vjvlin:ldlyapunov,stol:ldexp,sadiqgdv:pseudolog}, but we include
it for completeness.

Throughout this work, we denote by $(\Omega,\mathcal{F},\mathbb{P})$ a
common probability space that supports all defined random variables
and processes. Fix an integer $T > 0$, and consider a sequence of
(independent) queueing systems indexed by $n = 1,2,\ldots$, each with
its own arrival and channel state processes, and evolving as described
in Section \ref{sec:notation}. Henceforth, we explicitly reference by
the superscript $(n)$ any quantity associated with the $n$th
system. For any (possibly vector-valued) random process $X^{(n)}(k)$,
$k = 0,1,2,\ldots$ in the $n$th system, let us define its scaled (by
$1/n$), shifted and piecewise linear version $x^{(n)}(\cdot)$ on the
interval $[0,T]$ as follows: 
\[ x^{(n)}(t) = \left\{ \begin{array}{l}
    \frac{X^{(n)}(nt)}{n}, \quad \mbox{$nt$ an integer};\\
    \frac{X^{(n)}(\lfloor nt \rfloor)}{n} + \frac{X^{(n)}(\lceil nt \rceil) - X^{(n)}(\lfloor nt \rfloor)}{n(nt - \lfloor nt \rfloor)}, \\
    \mbox{otherwise}.
  \end{array}  \right. \]
In other words, we transform the discrete-time process
$X^{(n)}(\cdot)$ on $0,1,2,\ldots,nT$ to the piecewise linear and
continuous process $x^{(n)}(\cdot)$ on $[0,T]$ by (a) compressing time by
a factor of $n$, (b) scaling space by $\frac{1}{n}$
and (c) finally linearly interpolating between the discrete points. 

For the $n$th queueing system, with $k$ a nonnegative integer, we
define the following random processes central to our study of the
evolution of the system: 

\begin{itemize}
\item $F_i^{(n)}(k)$: The total number of packets to queue $i$ that
  arrived by time slot $k$,

\item $\hat{F}_i^{(n)}(k)$: The total number of packets that were
  served from queue $i$ by time slot $k$,

\item $C_\alpha^{(n)}(k)$: The total number of time slots before $k$
  when the observable subset $\alpha$ was chosen by the scheduling
  algorithm,

\item {\bf (Sub-state)} $R_\alpha^{(n)}(k)$: The vector of
  instantaneous service rates $R^{(n)}(k)$ restricted to the
  coordinates of $\alpha$, i.e., $R_\alpha^{(n)}(k) =
  (R_i^{(n)}(k))_{i \in \alpha}$,

\item $G_r^{\alpha,(n)}(k)$: The total number of time slots before
  time slot $k$ when the subset $\alpha$ was picked and its sub-state
  was $r$,

\item $\hat{G}_{ri}^{\alpha,(n)}(k)$: The number of time slots before
  time $k$ when subset $\alpha$ was picked, its observed sub-state was
  $r$ and queue $i \in \alpha$ was ultimately scheduled for service,

\item $Q_i^{(n)}(k)$: The length of queue $i$ at time slot $k$, whose
  evolution is specified in Section \ref{sec:notation},

\item $M^{(n)}(k)$: The (vector-valued) partial sums process
  corresponding to the sampled rates $R^{(n)}(k) \delta_{S(k)}$, i.e.,
  $M^{(n)}(k) \bydef \sum_{j=0}^k R^{(n)}(j) \delta_{S(j)}$. (Here,
  $\delta_S$ denotes the indicator vector of the subset $S$.)
\end{itemize} 

For right-continuous, non-decreasing functions $u: \mathbb{R} \to
\mathbb{R}$ and $v: \mathbb{R} \to \mathbb{R}$, we overload notation
and denote by $u$ and $v$ their respective induced Stieltjes measures
on $\mathbb{R}$, whenever the context is understood. Furthermore, when
$v \ll u$ (i.e., when $dv$ is absolutely continuous wrt $du$), we
denote by $\frac{dv}{du}$ the Radon-Nikodym derivative\footnote{The
  Radon-Nikodym derivative $\frac{dv}{du}$ is uniquely defined
  $du$-a.e.} of $v$ wrt $u$.

Suppose a sequence of scaled processes $f_i^{(n)}(\cdot)$,
$\hat{f}_i^{(n)}(\cdot)$, $c_\alpha^{(n)}(\cdot)$,
$g_r^{\alpha,(n)}(\cdot)$, $\hat{g}_{ri}^{\alpha,(n)}(\cdot)$,
$q_i^{(n)}(\cdot)$ and $m^{(n)}(\cdot)$ converges uniformly (over
$[0,T]$) to the corresponding ``limit functions'' $f_i(\cdot)$,
$\hat{f}_i(\cdot)$, $c_\alpha(\cdot)$, $g^\alpha_r(\cdot)$,
$\hat{g}^\alpha_{ri}(\cdot)$, $q_i(\cdot)$ and $m(\cdot)$ on $[0,T]$.
We call any such collection of joint limit functions, obtained via
appropriately scaled pre-limit sample paths, a \emph{Fluid Sample Path
  (FSP)} (we use the superscript $T$ to emphasize the finite horizon
$[0,T]$ if desired). We note that fluid sample paths inherit
Lipschitz continuity (with the same Lipschitz constant) from their
corresponding pre-limit processes indexed by $n$ (when the pre-limits
are Lipschitz-continuous), and are thus differentiable almost
everywhere. 

\noindent {\bf Note 1.} Wherever a scheduling algorithm is being
explicitly considered, we will use the term {\em valid FSP} to denote
an FSP that occurs with positive probability under the scheduling
algorithm.

\noindent {\bf Note 2.} We use $\dot{f}$ and $f'$ interchangeably, in
the paper, to denote the derivative of a (differentiable) function
$f$.

\section{Analysis: Singleton Subsets and Max-Queue}
\label{sec:singleton}
We first treat the simpler setting where the disjoint observable
subsets are all the singleton users in the system, i.e., $\mathcal{O}
= \{\{i\}: 1 \leq i \leq N\}$. We use the subscript $i$ to refer to
subsets $\alpha$. Thus, scheduling algorithms essentially become
sampling algorithms -- Step 2 of the algorithm is to schedule the lone
user whose channel state is observed. In what follows, we describe the
three key steps involved in showing that Max-Queue yields the optimal
decay rate of buffer overflow probability.

\subsection{Lower-bound for Max-Queue's Decay Rate}
Consider the queueing system operating under an arbitrary
\emph{nonrandom} scheduling algorithm, i.e., the algorithm's choice of
a singleton user in the current time slot is a deterministic function
of the entire history of observed users' indices and channel states,
and does not depend on the unobserved channel states in the
past\footnote{In formal terms, an arbitrary scheduling algorithm is to
  be understood as any map that is based on the current (or any finite
  past history) of system state.}.  Max-Queue with deterministic
tie-breaking (e.g., pick the lowest-indexed queue when there are two
or more longest queues) is an example of a nonrandom scheduling
algorithm, since the current user chosen depends on accumulated queue
lengths, which in turn depend directly on the channel rates obtained
as a result of past scheduling choices.

The sequence of observed users and their channel states under a
nonrandom scheduling algorithm is an outcome of sampling an iid
vector-valued process (i.e., the full channel state) in a nonrandom
and \emph{predictable} (i.e., with sampling indices depending only on
past observed history) manner. Our first key result (Proposition
\ref{prop:lbT}) essentially furnishes an upper bound for the deviation
probability of the queue-length process (equivalently the cumulative
process of observed channel states) in time slots $0, \ldots, nT$, in
terms of a novel sample-path large deviations rate function of the
user selection and channel state paths.

Let us fix $T > 0$. For $q_0 \in \mathbb{R}^N$, let
$\mathbb{P}_{q_0}^{n,T}$ be the probability measure of the $n$-th queueing
system conditioned on starting the system at $Q^{(n)}(0) = nq_0$ (i.e.
$q^{(n)}(0) = q_0$). If we denote by
$\mathcal{C}^+_\mathcal{L}([0,T])$ the space of nonnegative
$\mathbb{R}^N$-valued Lipschitz functions on $[0,T]$ equipped with
the supremum norm, then we have:

\begin{proposition}[Large Deviation Bound for a Finite Horizon]
\label{prop:lbT}
Let $\Gamma$ be a closed set of trajectories in
$\mathcal{C}^+_\mathcal{L}([0,T])$. Then, under any nonrandom
scheduling policy,
\begin{align}
-&\limsup_{n \to \infty} \frac{1}{n} \log \mathbb{P}_0^{n,T}\left[q^{(n)} \in \Gamma\right] \nonumber\\
\geq &\inf_{(m^T,c^T,q^T)} \quad \quad \int_0^T \left[\sum_{i=1}^N \dot{c_i}(t) \Lambda_i^* \left(\frac{dm_i}{dc_i}(t)\right)\right] dt \label{eqn:infT} \\
&\emph{subject to} \quad (m^T,c^T,q^T) \; \emph{a valid FSP}, \nonumber \\
&\hspace{2.2cm} q^T(0) = 0, q^T \in \Gamma, \nonumber 
\end{align}
with $\Lambda_i^*(\cdot)$ being the Legendre-Fenchel dual of
$\Lambda_i(\lambda) = \log \mathbb{E}[e^{\lambda R_i(0)}]$, i.e., the
Cram\'{e}r rate function for the empirical mean of the marginal rate
$(R_i(k))_{k}$.
\end{proposition}
Proposition \ref{prop:lbT} states that the ``correct'' sample-path
large deviations rate function, for algorithms that can sample only
singleton subsets of channels, is a combination of the standard rate
functions $\Lambda^*_i$ for the empirical means of individual channel
rates {\em weighted by the corresponding channel selection
  frequencies} $\dot{c}_i$. Note the crucial dependence of the rate
function on the subset selection process, captured by weighting
$\Lambda_i^*$ by $\dot{c}_i$ in (\ref{eqn:infT}) -- a significant
departure from the rate function studied for the standard case of full
channel state information where there is no pre-weighting by the
algorithm-dependent factor $\dot{c}$
\cite{vjvlin:ldlyapunov,stol:ldexp}. 

The proof of the proposition, presented in Appendix \ref{app:lbT},
relies on the key fact that the sample-path trajectory of any
nonrandom scheduling/sampling algorithm is completely determined by
only the \emph{sampled} user's index and the observed channel state at
all times, instead of the entire joint channel state process with
unobserved channel states. Also, since only one component of the joint
channel state is used at each instant, there is no loss of generality
in assuming that all the channel state processes are independent with
the original marginals. These two properties, together with
exchangeability of the channel state process, allow us to derive a
large deviations rate function for the random process of sampled
channel states, which is further transformed to the rate function
(\ref{eqn:infT}) as a function of empirical channel means and sampling
frequencies.

Having established a lower bound for the large deviations rate
function for the probability of queue overflow for a finite horizon
$T$ conditioned on a fixed starting state (Proposition
\ref{prop:lbT}), we now proceed to extend this result to the queue
overflow rate function for the \emph{stationary distribution} under
Max-Queue. Recall that a unique stationary distribution exists since
Max-Queue makes the irreducible and aperiodic system state Markov
chain positive recurrent \cite{gopcarsha12:subsets}. Intuitively, we
expect that the finite horizon probability distribution
$\mathbb{P}_{q_0}^{n,T}$ somehow ``tends'' to the stationary
distribution $\mathbb{P}$. Thus, we show that minimizing the right
hand side of (\ref{eqn:infT}) over all finite horizons $T > 0$ yields
a lower bound on this stationary overflow probability.

Such a procedure to extend finite horizon bounds to bounds on the
stationary probabilities has been developed earlier, using techniques
from Friedlin-Wentzell large-deviations theory
\cite{vjvlin:ldlyapunov,stol:ldexp}. A similar approach works in our
case, and for the sake of clarity we show only the crucial properties
for our model that are needed to obtain the result.

\begin{proposition}[Large Deviation Bound for the Stationary Distribution]
\label{prop:stationarylb}
Let $\mathbb{P}$ denote the stationary probability distribution of the
system state under the Max-Queue scheduling algorithm. Then,
\begin{align}
-&\limsup_{n \to \infty} \frac{1}{n} \log \mathbb{P}\left[||q^{(n)}(0)||_\infty \geq 1 \right] \nonumber \\ 
\geq &\inf_{T, (m^T,c^T,q^T)} \quad \quad \int_0^T \left[\sum_{i=1}^N \dot{c_i}(t) \Lambda_i^* \left(\frac{d{m_i}}{d{c_i}}(t)\right)\right] dt \nonumber \\
&\emph{subject to} \quad (m^T,c^T,q^T) \; \emph{a valid FSP in } [0,T], \nonumber \\
&\hspace{2.2cm} q^T(0) = 0, ||q^T(T)||_\infty \geq 1, \nonumber \\
&\hspace{2.2cm} T \geq 0. \label{eqn:stationarylb}
\end{align}
\end{proposition}

The reader is referred to Appendix \ref{app:stationarylb} for the
proof details. 

Applying Proposition \ref{prop:lbT} with $\Gamma = \{q \in
\mathcal{C}^+_\mathcal{L}([0,T]): ||q(T)||_\infty \geq 1 \}$ gives a
finite-horizon lower bound for the rate function of longest-queue
overflow. For any FSP $(m^T,c^T,q^T)$ feasible in the RHS of
(\ref{eqn:stationarylb}), we have
\begin{align*}
 \int_0^T &\left[\sum_{i=1}^N \dot{c_i}(t) \Lambda_i^* \left(\frac{d{m_i}}{d{c_i}}(t)\right)\right] dt 
\geq \inf_{t \in \mathcal{B}} \frac{\sum_{i=1}^N \dot{c_i}(t) \Lambda_i^* \left(\frac{d{m_i}}{d{c_i}}(t)\right)}{\frac{d}{dt} ||q(t)||_\infty},
\end{align*}
with $\mathcal{B}$ denoting the (almost all) points in $[0,T]$ at
which all the relevant derivatives exist. Let us define
\begin{align*}
  J_{*} \bydef \inf_{\substack{T \geq 0,\\(m^T,c^T,q^T),\\0 \leq t \leq T}}  \frac{\sum_{i=1}^N \dot{c_i}(t) \Lambda_i^* \left(\frac{d{m_i}}{d{c_i}}(t)\right)}{\frac{d}{dt} ||q(t)||_\infty},
\end{align*}
with the infimum over all FSPs $(m^T,c^T,q^T)$ feasible for
(\ref{eqn:stationarylb}), all \emph{regular points} $t$, and \emph{all
  finite horizons $T$}. This results in the following (weaker) lower
bound on the rate function of Max-Queue's stationary queue overflow
probability:

\begin{proposition}[Lower bound for Max-Queue's Queue Overflow Rate Function]
\label{prop:lbmaxq}
\begin{align}
-&\limsup_{n \to \infty} \frac{1}{n} \log \mathbb{P}\left[||q^{(n)}(0)||_\infty \geq 1 \right] \geq J_{*}.\label{eqn:lbmaxq}
\end{align}
\end{proposition}
Proposition \ref{prop:lbmaxq} is thus a ``cost per unit max-queue
drift'' lower bound on the decay rate of the queue overflow
probability under Max-Queue.

\subsection{Universal Large Deviations Upper Bound}
We next derive a \emph{uniform} upper bound for the stationary buffer
overflow probability decay rate, over all singleton-CSI scheduling
algorithms. A popular approach followed in recent work
\cite{stol:ldexp,vjvlin:ldlyapunov,sadiqgdv:pseudolog} to do this is
by estimating the cost of ``straight-line'' joint channel state sample
paths that universally cause buffer overflow. However, when only a
\emph{dynamically selected portion of the channel state} is visible to
the scheduling algorithm, the cost (\ref{eqn:infT}) of such
straight-line paths depends explicitly on the algorithm's sampling
behavior, so the standard approach fails.

For every $i$, let $\phi_i \geq 0$ denote a ``twisted'' mean rate for
channel $i$, and consider the quantity $\frac{\sum_{i}
  c_i'\Lambda_i^*(\phi_i)}{[\max_{i} \left(\lambda_i -
    c_i'\phi_i\right)]^+ }$. Here, we assume that $\sum_i c_i' = 1$,
and that the fraction is $\infty$ whenever the denominator is
$0$. Suppose a scheduling policy samples each channel $i$ with
frequency $c_i'$. Then, (a) the numerator of the above expression
corresponds to the ``instantaneous large deviations cost'' of
witnessing each channel $i$'s mean rate be $\phi_i$ (by
(\ref{eqn:infT})), while (b) the denominator can be interpreted as the
average rate with which the longest queue grows when each channel $i$
is sampled with a frequency $c_i'$. Maximizing the expression over all
possible user sampling/scheduling frequencies $\{c_i': \sum_i c_i' =
1, c_i' \geq 0 \}$ induced by scheduling algorithms should thus give
the highest possible large deviations cost for buffer overflow. This
intuition is formalized in the following key result:

\begin{proposition}[Universal Upper Bound on Decay Rate for any Algorithm]
\label{prop:ubany}
Let $\pi$ be a stabilizing scheduling policy\footnote{By a stabilizing
  scheduling policy $\pi$, we mean a scheduling rule that operates
  under the scheduling model described in Section \ref{sec:notation},
  and which makes the discrete time Markov chain of queue lengths
  aperiodic, irreducible and positive recurrent.} for the arrival rate
$\lambda = (\lambda_1, \ldots, \lambda_N)$, and let $\mathbb{P}^\pi$
be its associated stationary measure. For any $\phi_i \in
\mathbb{R}^+$, $i = 1, \ldots, N$,
\begin{align}
-\liminf_{n \to \infty} &\frac{1}{n} \log \mathbb{P}^\pi\left[||q^{(n)}(0)||_\infty \geq 1\right] \leq \sup_{\substack{\sum_i c_i' = 1\\c_i' \geq 0}} \frac{\sum_{i} c_i'\Lambda_i^*(\phi_i)}{[\max_{i} \left(\lambda_i - c_i'\phi_i\right)]^+}. \label{eqn:liminf}
 \end{align}
\end{proposition}

\emph{Note:} Each choice of the twisted means $(\phi_i)_i$ above
yields such an upper bound on the decay rate. Thus, the best possible
upper bound is obtained by minimizing (\ref{eqn:liminf}) over all
choices $(\phi_i)_i$.

According to Proposition \ref{prop:ubany}, an upper bound on the
buffer overflow rate function when scheduling with partial channel
observability is the largest ``weighted-cost per unit increase of the
maximum queue,'' over all possible frequencies of sampling subsets of
channels. We emphasize that the maximization over the sampling
frequencies $c_i'$, in (\ref{eqn:liminf}), is a distinct feature that
emerges while considering partial information algorithms, as opposed
to the case where scheduling is performed with full joint CSI. 

We refer the reader to Appendix \ref{app:ubany} for the proof of
Proposition \ref{prop:ubany}.  At the heart of the proof of
Proposition \ref{prop:ubany} is a twisted measure construction where
each channel's marginal rate is $\phi_i$. Observing that the
cumulative fluid service process $m(\cdot)$ is a submartingale under
the twisted measure for any scheduling algorithm, the Doob-Meyer
decomposition \cite{dur05:probtebook} allows us to express $m(\cdot)$
as the predictable algorithm-dependent component $\phi_i c_i(\cdot)$
plus a martingale noise component $\bar{m}(\cdot)$. This shows that
with high probability, the service provided to each queue $i$ is
approximated by $\phi_i c_i(\cdot)$, i.e., we can effectively treat
each channel $i$ as having a deterministic fluid service rate of
$\phi_i$. Analyzing this deterministic fluid system for overflow and
translating the results back to the original probabilistic system
gives us the result.

\subsection{Large Deviations Optimality of the Max-Queue Policy:
  Connecting the Upper and Lower Bounds}
The final step in the proof of optimality of Max-Queue (Corollary
\ref{cor:overallmaxqueue}) is carried out by showing that the lower
bound for Max-Queue (\ref{eqn:lbmaxq}) in fact dominates the uniform
upper bound (\ref{eqn:liminf}) over all scheduling policies:
\begin{proposition}[Matching Large Deviations Bounds, Max-Queue,
  Singleton Subsets]
  \label{prop:mqopt}
  There exist nonnegative $\hat{\phi}_1,\ldots,\hat{\phi}_N$, with $\lambda
  \notin \mathcal{C}(\hat{\phi}_1,\ldots,\hat{\phi}_N)$, such that
  \begin{align*}
    \sup_{\substack{\sum_i c_i' = 1\\c_i' \geq 0}} \frac{\sum_{i}
      c_i'\Lambda_i^*(\hat{\phi}_i)}{\left[\max_{i} \left(\lambda_i -
        c_i'\hat{\phi}_i\right)\right]^+} \leq J_{*}.
  \end{align*}
\end{proposition}

The proof of this result involves solving the non-convex problem for
the rate function lower bound given in Proposition \ref{prop:lbmaxq},
and relating the solution to a suitable uniform upper bound of the
type prescribed by Proposition \ref{prop:ubany}. It utilizes the
convexity and lower-semicontinuity of the rate functions
$\Lambda_i^*$, and is accomplished by considering the properties of
the $(\phi_i)_i$ which minimize the upper bound
(\ref{eqn:liminf}). The full proof appears in Appendix
\ref{app:mqopt}.

\section{Analysis: General Subsets and Max-Exp}
\label{sec:general}
In this section, we extend the queue overflow optimality result for
Max-Queue to the general setting of arbitrary disjoint subsets of
observable channels and the Max-Exp scheduling algorithm. For this, we
follow the same key steps in obtaining the Max-Queue result -- (a)
prove lower bounds on the buffer overflow exponent for Max-Exp, (b)
derive universal upper bounds on the buffer overflow exponent across
all scheduling algorithms using subset channel state information, and
(c) demonstrate that the upper and lower bounds match.

However, the approach to show optimality of the Max-Exp algorithm
warrants a more sophisticated analysis as compared to that of
Max-Queue. This is primarily due to the fact that the Max-Exp
algorithm is not a \emph{scaling-invariant} scheduling algorithm,
i.e., scaling all queue-lengths by a uniform constant changes the
scheduling behavior.  Intrinsically, Max-Exp operates at the
$O(\sqrt{n})$ time-scale, i.e., when all the queue lengths are $O(n)$,
a $O(\sqrt{n})$ change in them causes a shift in Max-Exp's scheduling
behavior. In other words, examining Max-Exp's scheduling over $O(n)$
time slot intervals effectively ``washes out'' information about its
actions, resulting in crude bounds. This sets Max-Exp apart from
Max-Queue which is naturally coupled to the timescale of $O(n)$ time
slots, and prevents us from using the standard $O(n)$ fluid scaling to
analyze the fluid sample path behavior of Max-Exp.

Hence, our analysis for Max-Exp proceeds by looking at sample paths of
the system's processes over intervals of $O(\sqrt{n})$ time slots. For
Step (a) above, analogous to Proposition \ref{prop:lbT}, we establish
a ``refined'' Mogulskii-type theorem for sample-path large deviations
of predictably sampled processes over a sub-$O(n)$ timescale (a
corresponding result for the full-CSI case was first proved in
\cite{stol:ldexp}).  Next, we use the framework of \emph{Local Fluid
  Sample Paths} (LFSPs, introduced in \cite{shasto01}) to obtain a
lower bound on the decay exponent of Max-Exp's overflow
probability. LFSPs allow us to ``magnify'' the standard $O(n)$ fluid
limit processes to examine events on the $O(\sqrt{n})$ ``local fluid''
timescale, and this helps us match the lower and upper bounds for the
decay exponent to establish the optimality of Max-Exp.

\subsection{Lower Bounding Max-Exp's Decay Rate: Refined-timescale
  Large Deviations for Sampled Processes and Local FSPs}
Here, we extend the sampling-based large-deviations bound from
Proposition \ref{prop:lbT} to hold over a finer-than-$O(n)$ timescale.
The basic idea here is to lower-bound the large deviations cost from
(\ref{eqn:infT}) by linearizing sample paths over the finer
timescale. This expresses the intuitive notion that over the finer
timescale, typical large deviations of random processes occur
``locally along straight lines''.

The general approach for studying scheduling behavior on
finer-than-$O(n)$ timescales is to introduce a positive integer
function $u(n)$, such that $u(n) \to \infty$ and $u(n)/n \to 0$ as $n
\to \infty$ (see Stolyar \cite{stol:ldexp}). We take $u(n) = \lceil\sqrt{n} \rceil$, which is the relevant
timescale for the dynamics of the Max-Exp scheduling rule
(\ref{alg:maxexp}).

For our analysis of the queue overflow rate function, we will need to
use this idea, along with the following {\em variable time
  discretization} for each observable subset. For any non-decreasing,
right-continuous-with-left-limits (RCLL) scalar function $h$ on
$[0,\infty)$, and any non-decreasing continuous function $\chi:
[0,\infty) \to [0, \infty)$, let $U^n_\chi h$ denote the continuous
and piecewise-linearized (according to $\chi$) version of $h$
constructed as follows: we divide $[0,\infty)$ into the contiguous
subintervals $[0,\chi(u(n)/n)]$, $[\chi(u(n)/n), \chi(2u(n)/n)]$,
$[\chi(2u(n)/n), \chi(3u(n)/n)]$, $\ldots$, and linearize $h$ between
its endpoints in each subinterval. For $t \geq 0$, let
$\theta^{(n)}(t)$ be the largest right-endpoint of a sub-interval that
does not exceed $t$. When the functions $h$ and $\chi$ are
vector-valued of (the same) finite dimension, we employ the same
notation $U^n_\chi h$ to mean the above linearization performed for
each of the individual scalar component functions in $h$ and its
counterpart function in $\chi$. In this case, the definition of
$\theta^{(n)}(t)$ is similarly extended in a component-wise fashion.

For each observable subset $\alpha$, let $\Lambda^*_\alpha$ be the
Sanov rate function \cite{dembo} for the empirical marginal
distribution of the state of its channels $(R_i(1))_{i \in
  \alpha}$. The domain of $\Lambda^*_\alpha$ is the
$|\mathcal{R}_\alpha|$-dimensional simplex where $\mathcal{R}_\alpha$
is the set of all possible sub-states for subset $\alpha$.

{\em Sampled Trace of the Queueing System:} We define here a random
object crucial to the analysis of Max-Exp. Consider the evolution of
the $n$-th queueing system in the time slots $1, 2, \ldots, nT$, and
suppose that subset $\alpha$ is picked by the scheduling algorithm
precisely at time slots $K_{\alpha}(1), K_{\alpha}(2), \ldots,
K_{\alpha}(C_{\alpha}(nT)) \in \{1,2,\ldots,nT\}$ .  Recall that
$\mathcal{R}_{\alpha}$ is the set of all possible sub-states wrt
subset $\alpha$. For each such sub-state $r \in \mathcal{R}_{\alpha}$,
we will find it convenient to associate it with the unit vector
$\mathbf{e}_r$ which is simply the
$|\mathcal{R}_{\alpha}|$-dimensional vector with $1$ in the $r$-th
position (according to a fixed ordering) and zeros everywhere else.

For each subset $\alpha$, set \[ V^{\alpha,(n)} \equiv V^{\alpha}
\bydef (\mathbf{e}_{R_{\alpha}(K_{\alpha}(1))},
\mathbf{e}_{R_{\alpha}(K_{\alpha}(2))},
\mathbf{e}_{R_{\alpha}(K_{\alpha}(3))},
\ldots,\mathbf{e}_{R_{\alpha}(K_{\alpha}(C_{\alpha}(nT)))})\] i.e.,
the $j$-th element of $V^{\alpha}$ simply records what sub-state was
sampled when $\alpha$ was picked for the $j$-th time $K_{\alpha}(j)$.

We call $V^{(n)} \equiv V \bydef (V^{\alpha})_{\alpha \in
  \mathcal{O}}$ the {\em sampled trace} of the queueing system. The
sampled trace represents, in words, the sequence of sub-state
observations seen by the scheduling algorithm, organized according to
the subsets sampled during the operation of the scheduling
algorithm. Note also that for any {\em deterministic} scheduling
algorithm, the sampled trace completely specifies the entire sample
path of the queue lengths (in conjunction with the arrival sequence
which is assumed to be deterministic).

Corresponding to each possible sampled trace $V$, we define its
partial sums process \[ W^{\alpha,(n)}(k) \equiv W^{\alpha} =
\sum_{j=1}^k V^{\alpha}(j), \quad 1 \leq k \leq C_{\alpha}(nT) \] for
each observable subset $\alpha \in \mathcal{O}$. We then define
$W^{(n)} \equiv W \bydef (W^{\alpha})_{\alpha \in \mathcal{O}}$. Note
that each sampled trace $V$ corresponds bijectively to its
partial-sums process $W$. Also, as per convention, we use $w \equiv
w^{(n)}$ and $v \equiv v^{(n)}$ to denote the rescaled (by $n$)
versions of $W$ and $V$ respectively. 

Let us define the candidate sample-path large deviations rate function for
our queueing system as follows:
\begin{align*}
 \hat{J}_t(z,s) \bydef &\int_0^{z_{\alpha}(t)} \; \sum_{\alpha \in
  \mathcal{O}} \Lambda^*_{\alpha}\left(s'_{\alpha}(u)\right) du, \\
&s_{\alpha} \in \mathcal{A} \left([0,T] \to
  \mathbb{R}^{|\mathcal{R}_\alpha|}\right), \\
&z_{\alpha} \in \mathcal{A} \left([0,T] \to \mathcal{R} \right), \; \; t \in [0,T].
\end{align*}
Here, $\mathcal{A}$ is used to denote the set of absolutely continuous
functions.

In order to track large deviations costs over the refined $u(n)$
timescale, let us introduce the notion of a Generalized Fluid Sample
Path (GFSP) \cite{stol:ldexp}, built upon the framework of standard
FSPs.

\begin{definition}[Generalized Fluid Sample Path (GFSP)]
  Suppose that there exists an increasing subsequence $\{n\}$ of the
  sequence of positive integers such that
  \begin{enumerate}
  \item For each $n$, there is a valid realization $(f^{(n)},
    \hat{f}^{(n)}, c^{(n)}, g^{(n)}, \hat{g}^{(n)}, q^{(n)}, m^{(n)},
    w^{(n)})$.

  \item As $n \to \infty$, we have the u.o.c. convergence 
    \[ (f^{(n)}, \hat{f}^{(n)}, c^{(n)}, g^{(n)}, \hat{g}^{(n)},
    q^{(n)}, m^{(n)}, w^{(n)}) \to (f, \hat{f}, c, g, \hat{g}, q, m,
    w) \] for a set of limiting, Lipschitz continuous functions $(f,
    \hat{f}, c, g, \hat{g}, q, m, w)$, and the u.o.c. convergence
    \[ \bar{J}^{(n)} \equiv (\bar{J}_t^{(n)}, t \in [0,T]) \bydef
    \left(\hat{J}_t\left(c^{(n)},U^n_{c^{(n)}} w^{(n)}\right), t \in
      [0,T]\right) \; \to \; \bar{J} = (\bar{J}_t, t \in [0,T]) \]
    for a non-negative non-decreasing Lipschitz-continuous function
    $\bar{J}$.

    Then, the entire construction 
    \[ [\{n\}; (f^{(n)}, \hat{f}^{(n)}, c^{(n)}, g^{(n)},
    \hat{g}^{(n)}, q^{(n)}, m^{(n)}, w^{(n)}), \bar{J}^{(n)}; (f,
    \hat{f}, c, g, \hat{g}, q, m, w), \bar{J}] \] is called a
    generalized fluid sample path (GFSP). The non-decreasing function
    $\bar{J}$ will be called the {\em refined cost function} of the
    GFSP.
  \end{enumerate}
\end{definition}

We note that for any $0 \leq t_1 < t_2 < \infty$, 
\begin{equation}
 \bar{J}_{t_2} - \bar{J}_{t_1} \geq \hat{J}_{t_2}(c,w) -
\hat{J}_{t_1}(c,w), \label{eqn:convexity1}
\end{equation} as a result of convexity of the
$\Lambda^*_\alpha$, $\alpha \in \mathcal{O}$, and Jensen's inequality.

The following finite-horizon result strengthens Proposition
\ref{prop:lbT}. It states that for any \emph{nonrandom} scheduling
algorithm, the sample path large deviations rate function for the
queue length process is lower-bounded by the minimum \emph{refined
  cost} over valid GFSPs.

\begin{proposition}[Refined-time-scale Lower Bound on Large Deviation
  Rate Function]
  \label{prop:lbTref}
  Let $\Gamma$ be a closed set of trajectories in
  $\mathcal{C}^+_\mathcal{L}([0,T])$. Then, under a nonrandom
  scheduling policy,
  \begin{align}
    -\limsup_{n \to \infty} \frac{1}{n}  \log&
    \mathbb{P}_0^{n,T}\left[q^{(n)} \in \Gamma\right]  \geq \nonumber \\
    & \inf  \left\{\bar{J}_0 : \exists  \mbox{GFSP $\psi$ on} \;
      [0,T], \bar{J} \in \psi, q \in \psi, q \in \Gamma
    \right\}. \label{eqn:lbTref}
\end{align}
\end{proposition}

The proof appears below. The proof uses ideas from the large
deviations of sampling (in the manner of Proposition \ref{prop:lbT}),
the crucial concept of {\em sampled traces}, and a {\em variable
  discretization}-version of a refined Mogulskii theorem first shown
by Stolyar \cite{stol:ldexp}, in order to establish the rate function
bound (\ref{eqn:lbTref}).

\begin{proof}
  For an observable subset $\alpha$, let $\tilde{\mathbb{P}}_{\alpha}$
  be the probability measure on $\{\mathbf{e}_r: r \in
  \mathcal{R}_{\alpha}\}$ such that
  $\tilde{\mathbb{P}}_{\alpha}[\mathbf{e}_r] = \prob{R_{\alpha}(1) =
    r}$ $\forall r \in \mathcal{R}_{\alpha}$.  Form the ``marginal''
  product distribution for subset $\alpha$ as
  $\hat{\mathbb{P}}_{\alpha} \bydef \tilde{\mathbb{P}}_{\alpha} \times
  \tilde{\mathbb{P}}_{\alpha} \times \cdots$ (i.e., extend
  $\tilde{\mathbb{P}}_{\alpha}$ to countably infinite sequences in an
  iid fashion), and finally take the product of these marginal
  measures, across observable subsets, to get $\hat{\mathbb{P}} \bydef
  \prod_{\alpha \in \mathcal{O}} \hat{\mathbb{P}}_{\alpha}$.  For any
  candidate sampled trace $v = (v^{\alpha})_{\alpha \in \mathcal{O}}$,
  we understand $\hat{\mathbb{P}}[v]$ as $\prod_{\alpha \in
    \mathcal{O}} \hat{\mathbb{P}_\alpha}[v^{\alpha}] = \prod_{\alpha
    \in \mathcal{O}} \prod_j
  \tilde{\mathbb{P}_\alpha}[v^{\alpha}(j)]$.

  The lemma below states that for any {\em deterministic} scheduling
  algorithm, the probability distribution of the sampled trace of the
  queueing system is {\em identical} under both the original measure
  $\mathbb{P}_{q_0}^{nT}$ and the product-of-marginals measure
  $\hat{\mathbb{P}}$ defined above. This will subsequently allow us to
  apply sample-path large-deviations results on the iid measure
  $\hat{\mathbb{P}}$ instead of the more complex, correlated measure
  ${\mathbb{P}}_{q_0}^{nT}$.

  \begin{lemma}
    \label{lem:samptraceequiv}
    For every sampled trace $v = (v^{\alpha})_{\alpha \in
      \mathcal{O}}$, $\mathbb{P}_{q_0}^{nT}[V = v] =
    \hat{\mathbb{P}}[v]$.
  \end{lemma}

  Continuing with the proof of the proposition, let $\Gamma^{(n)}$
  (resp. $\mathsf{C}^{(n)}$) be the set of all valid rescaled sampled
  traces (resp. all valid rescaled subset-sampling trajectories
  $c^{(n)}$) in the $n$-th queueing system\footnote{Since sampled
    traces and their partial sums processes are in one-to-one
    correspondence, we take the liberty of referring to them
    interchangeably.} that, starting with initial queue lengths $q_0$,
  result in queue length sample paths belonging to $\Gamma$. We also
  let $\Gamma \mathsf{C}^{(n)} \bydef \{(w^{(n)},c^{(n)}): w^{(n)} \in
  \Gamma^{(n)}\}$ denote the set of valid (sampled trace, sample
  trajectory) pairs corresponding to sampled trace trajectories
  belonging to $\Gamma^{(n)}$. We have, due to Lemma
  \ref{lem:samptraceequiv},
  \begin{align*}
    \mathbb{P}_{q_0}^{nT}\left[w^{(n)} \in \Gamma^{(n)}\right] &=
    \sum_{w \in \Gamma^{(n)}} \mathbb{P}_{q_0}^{nT}\left[ W^{(n)} =
      w \right] \\
    &= \sum_{w \in \Gamma^{(n)}} \hat{\mathbb{P}}[w ]  \\
    &= \hat{\mathbb{P}}\left[\left\{w: w \in \Gamma^{(n)} \right\}
    \right],
  \end{align*}
  where the final step is due to the fact that the sampled trace
  uniquely specifies the queueing system's complete trajectory, and so
  sampled traces corresponding to different sample paths of the system
  must necessarily be different. Passing in this fashion to the iid
  measure $\hat{\mathbb{P}}$ allows us to use a refinement of
  Mogulskii's theorem \cite{dembo} first established by Stolyar
  \cite[Theorem 7.1]{stol:ldexp}, to estimate the large deviations
  rate function. As a consequence, we can write\footnote{\cite[Theorem
    7.1]{stol:ldexp} derives the rate function bound under a
    \emph{fixed} discretization of the time axis where the
    discretization rate is always unity, i.e., $\chi(x) = x$ $\forall
    x \geq 0$. When the discretization rate is variable and depends on
    the subset selection frequency $c^{(n)}$, it is not hard to see
    that the result of \cite[Theorem 7.1]{stol:ldexp} extends with
    $U^n h$ replaced by $U^n_{c^{(n)}} h$ -- the key property that
    affords us this extension is that the \emph{size} of each
    subinterval ($u(n)/n$ in \cite{stol:ldexp} and $\chi((k+1)u(n)/n)
    - \chi(ku(n)/n)$ here) does not matter; it is the \emph{number} of
    discretized intervals ($Tu(n)/n$ both in \cite{stol:ldexp} and
    here) that is the crucial ingredient in the bound in \cite[Theorem
    7.1]{stol:ldexp}.}
  \begin{align}
    \label{eqn:refmog}
    -\limsup_{n \to \infty} \frac{1}{n} \log
    \hat{\mathbb{P}}\left[\left\{w^{(n)}: w^{(n)} \in \Gamma^{(n)}
      \right\}\right] &\geq \liminf_{n \to \infty} \;
    \inf\left\{\hat{J}_{T}(c,U^n_c w): (w,c) \in  \Gamma \mathsf{C}^{(n)} \right\}.
  \end{align}
  Let the limit inferior on the right-hand side of (\ref{eqn:refmog})
  above be denoted by $\zeta$. It follows that we can find for each
  $n$ a $w_n \in \Gamma^{(n)}$ and $c_n = (c_{n,\alpha})_{\alpha \in
    \mathcal{O}} \in \mathbb{R}^{|\mathcal{O}|}$, such that $(w_n,c_n)
  \in \Gamma \mathsf{C}^{(n)}$ and $\hat{J}_{0}(c_n,U^n_{c_n} w_n) \to
  \zeta$. Using uniform Lipschitz continuity of the $\{w_n\}$ and
  $\{c_n\}$, we can extract a subsequence of trajectories $(w_n,
  c_n)$ which converges and forms a GFSP with refined cost $\zeta$,
  and which satisfies, by construction, the conditions on the
  right-hand side of (\ref{eqn:lbTref}). Thus, we get
  \[ \zeta \geq \inf \; \left\{\bar{J}_0 : \exists \; \mbox{GFSP
      $\psi$ on} \; [0,T], \bar{J} \in \psi, q \in \psi, q \in \Gamma
  \right\},\] completing the proof.
\end{proof}

Similar to extending the result of Proposition \ref{prop:lbT} to the
stationary queue length distribution, minimizing the RHS of
(\ref{eqn:lbTref}) across FSPs over all finite time horizons $T > 0$
yields a lower bound for the large deviations rate of the
\emph{stationary} queue overflow probability. This uses standard tools
(see, for instance, \cite{stol:ldexp,vjvlin:ldlyapunov}), and we omit
the proof for brevity.

\subsection{Extending The Lower Bound to The Stationary Queue
  Distribution}
As with the approach followed to extend the result of Proposition
\ref{prop:lbT} to the stationary measure under Max-Queue (i.e., to
Proposition \ref{prop:stationarylb}), we can use standard
Friedlin-Wentzell-type techniques to extend Proposition
\ref{prop:lbTref} to a large-deviations lower bound
\cite{vjvlin:ldlyapunov,stol:ldexp} for the stationary measure under
the Max-Exp scheduling policy. Note that this requires showing that
Max-Exp is throughput-optimal\footnote{Again, the
  throughput-optimality holds among all scheduling algorithms that
  base their decision on the current (or any finite past history of)
  system state.} -- a fact whose proof we omit for brevity, but which
results from a fairly straightforward modification of the proof of
throughput-optimality of the Max-Sum Queue algorithm (see
\cite{gopcarsha12:subsets} for details).

\begin{theorem}
\label{thm:stationarylbref}
Let $\mathbb{P}$ denote the stationary measure induced by the Max-Exp
policy. Then,
  \begin{align}
    -&\limsup_{n \to \infty} \frac{1}{n} \log \mathbb{P}\left[||q^{(n)}(0)||_\infty \geq 1\right] \nonumber\\
    \geq \; &\inf_{T \geq 0} \inf \left\{\bar{J}_t : \exists \;
      \mbox{GFSP $\psi$ on} \; [0,T], \bar{J} \in \psi, q \in \psi, t
      \in [0,T], q(0) = 0, ||q(t)||_\infty \geq 1
    \right\}. \label{eqn:stationarylbref}
\end{align}
\end{theorem}

\subsection{Straight-line Uniform LD Upper Bounds over all policies}

In this section, we establish a crucial \emph{upper} bound on decay
rate of the stationary queue-overflow probability \emph{uniformly} for
any stabilizing scheduling policy, along the lines of Proposition
\ref{prop:ubany}. This is stated and carried out in terms of
``twisted'' marginal probability distributions for the subset channel
states, and the local/subset-based throughput regions that they induce.

Recall that for an observable subset $\alpha$, $\mathcal{R}_\alpha$
denotes the (finite) set of all possible (joint) sub-states that can
be observed channels in $\alpha$. We use $\Pi_\alpha$ to denote the
$|\mathcal{R}_\alpha|$-valued simplex, i.e., the set of all
probability measures on the sub-states of $\alpha$. Any distribution
$\phi_\alpha \in \Pi_\alpha$ induces a \emph{subset throughput
  region} $V_{\phi_\alpha}$, which represents all the long-term
average service rates that can be sustained to users in $\alpha$ when
the sub-states are distributed as $\phi_\alpha$ (see also
\cite{stol:async,gopcarsha12:subsets}). The uniform large-deviations
upper-bound can now be stated for any stabilizing scheduling policy
$\pi$:
 
\begin{theorem}
  \label{thm:ubanysubset} 
  Let $\pi$ be a stabilizing scheduling policy for arrival rates
  $\lambda = (\lambda_1, \ldots, \lambda_n)$, and let $\mathbb{P}^\pi$
  be the associated stationary measure. Let distributions
  $\phi_\alpha \in \Pi_\alpha$ be fixed, for every $\alpha$, such
  that $\lambda \notin \mathcal{CH}((V_{\phi_\alpha})_\alpha)$. Then, 
\begin{align}
  -\liminf_{n \to \infty} \frac{1}{n} \log
  \mathbb{P}^\pi\left[||q^{(n)}(0)||_\infty \geq 1\right] \leq
  \sup_{\substack{\sum_\alpha c_\alpha' = 1\\c_\alpha' \geq 0}} \left[
    \frac{\sum_{\alpha}
      c_\alpha'\Lambda_\alpha^*(\phi_\alpha)}{\max_{\alpha, v_\alpha
        \in V_{\phi_\alpha}} \max_{i \in \alpha} (\lambda_i -
      c_\alpha' v_{\alpha,i})} \right]. \label{eqn:liminfsubset}
 \end{align} 
\end{theorem}

\subsection{Showing Max-Exp's Overflow Exponent is Optimal}
Finally, in this section, we establish that the large-deviations
buffer overflow exponent for the Max-Exp scheduling algorithm is in
fact optimal over all stabilizing scheduling rules\footnote{By
  optimal, we mean optimal among all stabilizing scheduling algorithms
  that base their decision on the current (or any finite past history
  of) system state.}. For this, we leverage the large-deviations lower
bound for the Max-Exp scheduling algorithm (Theorem
\ref{thm:stationarylbref}) and show that it is actually a uniform
upper bound over all scheduling rules as prescribed by Theorem
\ref{thm:ubanysubset}.

Our approach at the high level is comprised of the following steps:
\begin{enumerate}
\item Consider a feasible FSP $(q,\bar{J})$ on $[0,T]$ for Theorem
  \ref{thm:stationarylbref}, i.e., $q(0) = 0$, $q(t) = 1$ for some $t
  \in [0,T]$. We show, by ``magnifying'' the FSP about some $\tau \in
  [0,T]$ and taking ``local'' fluid limits, that the ``unit
  large-deviations cost'' of raising the longest queue in the
  associated Local Fluid Sample Path (LFSP)
  \cite{shasto01,stol:ldexp} at $\tau$ is close to the total
  FSP cost $\bar{J}_T$.
\item Thus, a further lower bound on the Max-Exp rate function is the
  least ``large-deviations cost per unit increase of longest queue''
  over all feasible local fluid sample paths -- call it $J_*$.
\item In the context of Theorem \ref{thm:ubanysubset}, we exhibit
  suitable twisted subset distributions $\phi_\alpha \in \Pi_\alpha$
  $\forall \alpha$ such that the RHS of (\ref{eqn:liminfsubset}) is at
  most $J_*$, proving the claimed result. \\
\end{enumerate}

\subsubsection{From Low Cost FSPs to Low Cost Local FSPs}
The variational problem on the right-hand side of
(\ref{eqn:stationarylbref}) necessitates a closer look at the
derivatives of fluid sample paths under the Max-Exp scheduling
algorithm. At the same time, since the Max-Exp rule naturally operates
at the $O(\sqrt{n})$ timescale, derivative information typically is
``washed out'' of the standard $O(n)$-scaled fluid sample paths. This
motivates us to define and use Local Fluid Sample Paths (LFSPs) with a
$O(\sqrt{n})$-type scaling, in which information about scheduling
choices and drifts can be clearly
understood with regard to the Max-Exp scheduling rule. \\

The formal LFSP construction is along the lines of that used in
\cite{shasto01,stol:ldexp}, and is as follows. Consider a
standard fluid sample path on $[0,T]$ (along with its prelimit
functions) and call it $\psi$. Let us introduce the ``recentered''
queue lengths
\[\tilde{Q}_i^{(n)}(t) \bydef Q_i^{(n)}(t) - b_i\sqrt{\bar{Q}^{(n)}(t)},  \]
where $b_i$, $i = 1, \ldots, N$ are such that for each observable
subset $\alpha$, the vector $(e^{b_i})_{i \in \alpha}$ is an outer
normal to the subset rate region $V_\alpha$ (under the natural
marginal distribution of the sub-state $R_\alpha(1)$) at some point
$v^*_\alpha \in V_\alpha$ such that $v^*_\alpha >
\lambda|_\alpha$. The fluid-scaled version of $\tilde{Q}_i^{(n)}$ is
\[\tilde{q}_i^{(n)}(t) = q_i^{(n)}(t) - \frac{b_i}{\sqrt{n}}\sqrt{\bar{q}^{(n)}(t)},  \]
so we have the uniform convergence
\[ \tilde{q}_i^{(n)} \to q_i, \] 
and
\[ \tilde{q}_*^{(n)} \bydef \max_i \tilde{q}_i^{(n)} \to q_* \bydef \max_i q_i. \]

Let $\tau \in [0,T]$ be fixed, such that $q_*(\tau) > 0$. Also, fix $S
> 0$ and set $\sigma_n \bydef
\frac{1}{\sqrt{n}}\sqrt{\bar{q}^{(n)}(\tau)}$. Suppose we pick a
sequence of time intervals $[t_1^{(n)}, t_2^{(n)}] \subseteq [0,T]$,
indexed by $n$, such that $t_2^{(n)}- t_1^{(n)} = S\sigma_n$ and
$t_1^{(n)} \to \tau$ as $n \to \infty$. Then, for each $n$ and $s \in
[0,S]$, consider the following ``centered'' and ``rescaled''
functions:
\begin{align*}
\di q_i^{(n)}(s) &\bydef \frac{1}{\sigma_n}[\tilde{q}_i^{(n)}(t_1^{(n)} + \sigma_n s) - \tilde{q}_*^{(n)}(t_1^{(n)})], \quad i = 1,\ldots, N, \\
\di q_*^{(n)}(s) &\bydef \max_i \di q_i^{(n)}(s) = \frac{1}{\sigma_n}[\tilde{q}_*^{(n)}(t_1^{(n)} + \sigma_n s) - \tilde{q}_*^{(n)}(t_1^{(n)})], \\
\di f_i^{(n)}(s) &\bydef \frac{1}{\sigma_n}[f_i^{(n)}(t_1^{(n)} + \sigma_n s) - f_i^{(n)}(t_1^{(n)})], \quad i = 1,\ldots,N, \\
\di \hat{f}_i^{(n)}(s) &\bydef \frac{1}{\sigma_n}[\hat{f}_i^{(n)}(t_1^{(n)} + \sigma_n s) - \hat{f}_i^{(n)}(t_1^{(n)})], \quad i = 1,\ldots,N, \\
\di c_\alpha^{(n)}(s) &\bydef \frac{1}{\sigma_n}[c_\alpha^{(n)}(t_1^{(n)} + \sigma_n s) - c_\alpha^{(n)}(t_1^{(n)})], \quad \alpha \in \mathcal{O}, \\
\di g_r^{\alpha,(n)}(s) &\bydef \frac{1}{\sigma_n}[g_r^{\alpha,(n)}(t_1^{(n)} + \sigma_n s) - g_r^{\alpha,(n)}(t_1^{(n)})], \quad \alpha \in \mathcal{O}, r \in \mathcal{R}_\alpha, \\
\di \hat{g}_{ri}^{\alpha,(n)}(s) &\bydef \frac{1}{\sigma_n}[\hat{g}_{ri}^{\alpha,(n)}(t_1^{(n)} + \sigma_n s) - \hat{g}_{ri}^{\alpha,(n)}(t_1^{(n)})], \quad \alpha \in \mathcal{O}, r \in \mathcal{R}_\alpha, i = 1,\ldots,N, \\
\di m^{(n)}(s) &\bydef \frac{1}{\sigma_n}[m^{(n)}(t_1^{(n)} + \sigma_n s) - m^{(n)}(t_1^{(n)})].
\end{align*}



It follows that we can choose a subsequence of $n$ along which the
following uniform convergence to Lipschitz functions holds on $[0,S]$
\cite{stol:ldexp}:
\begin{align}
&\left(\di q^{(n)}, \di q_*^{(n)}, \di f^{(n)}, \di \hat{f}^{(n)}, (\di c_\alpha^{(n)})_\alpha, (\di g_r^{\alpha,(n)})_{\alpha r}, (\di \hat{g}_{ri}^{\alpha,(n)})_{\alpha r,i}, \di m^{(n)}\right) \to \nonumber \\
& \quad \left(\di q, \di q_*, \di f, \di \hat{f}, (\di g_r^{\alpha})_{\alpha r}, (\di \hat{g}_{ri}^{\alpha})_{\alpha r,i}, \di m\right). \label{eqn:lfspdef}
\end{align}

Note that each $\di q_i$ can be either finite Lipschitz or $-\infty$;
we appropriately extend the definition of uniform convergence in the
latter case. We call the tuple on the right-hand side of
(\ref{eqn:lfspdef}) above a \emph{Local Fluid Sample Path} at (scaled)
time $\tau$.

We also have the following consequence of the (marginal) convexity of
$\hat{J}$, in close analogy with (\ref{eqn:convexity1}):
\begin{equation}
\label{eqn:convexity2}
\liminf_{n \to \infty} \frac{1}{\sigma_n} [\bar{J}^{(n)}_{t_2^{(n)}}
- \bar{J}^{(n)}_{t_1^{(n)}}] \geq \hat{J}_S\left(\di c, \frac{d \di g}{d
  \di c}\right).
\end{equation} 

The following key lemma, along the lines of Lemma 9.1 in
\cite{stol:ldexp}, is crucial to understand the local timescale
dynamics of the Max-Exp scheduling algorithm:

\begin{lemma}
  \label{lem:lfspdynamics}
  For any LFSP over an interval $[0,S]$, 
  \begin{enumerate}
  \item The following derivatives exist Lebesgue-a.e.\footnote{As
      convention, we take $\frac{d}{dt}(\infty) = 0$.} and are finite:
    \begin{align*}
      \di \dot{q}, \quad \di \dot{q}_*, & \quad \di \dot{c}_\alpha,
      \quad \di \dot{f}, \quad v \bydef \di
      \dot{\hat{f}}, \quad \di \dot{g}^\alpha_r, \quad \di \dot{\hat{g}}^\alpha_{ri}, \quad \di \dot{m}.
    \end{align*}
  \item For every $\alpha \in \mathcal{O}$ and $r \in
    \mathcal{R}_\alpha$, $\di g^\alpha_r \ll \di c_\alpha$ (wrt the
    corresponding Lebesgue-Stieltjes induced measures). Thus, there
    exists ($c_\alpha$-a.e.) a version of the Radon-Nikodym derivative
    $\phi_{\alpha r} \bydef \frac{d \di g^\alpha_r }{d \di
      c_\alpha}$.

  \item The following relations hold, whenever the relevant derivatives
    exist, for $s \in [0,S]$:
  \begin{align}
    \di \dot{f}(s)  &= \lambda, \label{eqn:property1} \\
    \di \dot{q}(s) &= \di \dot{f}(s) - v(s), \label{eqn:property2} \\
    \di q_*(s) &= \max_i \di q_i(s), \label{eqn:property3} \\
    \di \dot{q}_*(s) = \di \dot{q}_i(s) &\quad \mbox{for each $i$ such that} \; \di q_i(s) = \di q_*(s), \label{eqn:property4}\\
    v_i(s) = \sum_{r \in \mathcal{R}_\alpha} \di \dot{\hat{g}}^\alpha_{ri}(s) \mu^\alpha_{ri}& \quad \mbox{for each} \; i \in \alpha, \alpha \in \mathcal{O}, \label{eqn:property5} \\
    \sum_{i \in \alpha} \di \dot{\hat{g}}^\alpha_{ri}(s) &= \di \dot{g}^\alpha_{r}(s), \label{eqn:property6} \\
    \sum_{r \in \mathcal{R}_\alpha} \di \dot{g}^\alpha_{r}(s) &= \di \dot{c}_\alpha(s), \label{eqn:property7} \\
    \sum_{r \in \mathcal{R}_\alpha} \di \phi_{\alpha r}(s) &= 1, \label{eqn:property8} \\
    \sum_{\alpha \in \mathcal{O}} \di \dot{c_\alpha}(s) &= 1, \label{eqn:property9} \\
    \forall \alpha \in \mathcal{O} \quad v_\alpha(s) &= \di \dot{c}_\alpha(s) \times \arg \max_{\eta \in V_{\phi_\alpha}(s)} \ip{e^{\di q(s) + b}}{\eta}_\alpha, \label{eqn:property10} \\
    \sum_{i \in \beta} e^{\di q_i(s) + b_i} < \max_{\alpha \in \mathcal{O}} \sum_{i \in \alpha} e^{\di q_i(s) + b_i} \; &\Rightarrow \; \di \dot{c}_\beta(s) = 0 \quad \mbox{for each} \; \beta \in \mathcal{O}, \label{eqn:property11} \\
    \left. \frac{d}{du} \sum_{i \in \beta} e^{\di q_i(u) +
        b_i}\right|_{u=s} = &\left. \frac{d}{du} \sum_{i \in \gamma}
      e^{\di q_i(u) + b_i}\right|_{u=s} \nonumber \\
    \mbox{whenever} \; \beta, \gamma &\in \arg\max_{\alpha \in
      \mathcal{O}} \sum_{i \in \alpha} e^{\di q_i(s) + b_i},
    \quad \mbox{and} \nonumber \\
    \mbox{$s$ is a regular point of} \; \sum_{i \in \beta} e^{\di
      q_i(s) + b_i}&, \sum_{i \in \gamma} e^{\di q_i(s) + b_i} \;
    \mbox{and} \; \max_{\alpha \in \{\beta,\gamma\}} \sum_{i \in
      \alpha} e^{\di q_i(s) + b_i}. \label{eqn:property12}
  \end{align}
  \end{enumerate}
\end{lemma}

\begin{proof}
  The first assertion of the lemma follows due to the absolute
  continuity of the LFSP functions being considered, which, in turn,
  is a consequence of the corresponding Lipschitz-continuous prelimit
  functions.

  The second assertion of the lemma is due to the fact that $\di
  g_r^{\alpha,(n)} \ll \di c_\alpha^{(n)}$ for the prelimit functions
  -- a queue belonging to a subset cannot be scheduled without first
  choosing the subset.

  As regards the third assertion, properties
  (\ref{eqn:property1})-(\ref{eqn:property9}) follow due to the
  corresponding properties of their prelimit LFSP functions, together
  with the (Lebesgue-a.e.) derivatives. 

  Property (\ref{eqn:property10}) is a key property of the Exponential
  scheduling rule, and has been established previously in the work of
  Shakkottai and Stolyar \cite{shasto01}. It is a consequence of the
  ratios of the $\exp(\cdot)$ terms (for different queues $i$) in the
  definition of the intra-subset Exponential rule (i.e., Step (2) in
  Algorithm \ref{alg:maxexp}) converging to the ratios $e^{\di q_i(s)
    + b_i}$ on the LFSP time scale.
  
  Property (\ref{eqn:property11}) follows from Max-Exp's subset
  selection criterion (i.e., Step (1) in Algorithm \ref{alg:maxexp})
  applied to the prelimit LFSP functions, along with the convergence
  of the $\exp(\cdot)$ terms to the ratios $e^{\di q_i(s) + b_i}$ as
  noted above.

  To show (\ref{eqn:property12}), we argue by contradiction. If $s$ is
  a regular point of the subset functions $L_{\beta}(\cdot) \bydef
  \sum_{i \in \beta} e^{\di q_i(\cdot) + b_i}$ and $L_{\gamma}(\cdot)
  \bydef \sum_{i \in \gamma} e^{\di q_i(\cdot) + b_i}$, and
  $L_{\beta}(s) = L_{\gamma}(s)$ holds, but
  \[ \left. \frac{d}{du} L_{\beta}(s) \right|_{u=s} \neq
  \left. \frac{d}{du} L_{\gamma}(s) \right|_{u=s}, \] then a simple
  argument shows that the function
  $\max(L_{\beta}(\cdot),L_{\gamma}(\cdot))$ cannot be differentiable
  at $s$, yielding a contradiction.  \\
\end{proof}

With this framework of LFSPs set up, we can resume the main
development from Theorem \ref{thm:stationarylbref}. Consider a
feasible GFSP $\psi$ on $[0,T]$ for the right-hand side of
(\ref{eqn:stationarylbref}) (i.e., for which $q(0) = 0$ and
$||q(t)||_\infty = 1$ for some $t \in [0,T]$), and whose refined cost
is $\bar{J}_t$.  Fix an arbitrary $\epsilon > 0$. Then, there must
exist a time point $\tau \in (0,t)$ such that $q_*(\tau) > 0$,
$q'_*(\tau) > 0$, $\bar{J}'_\tau > 0$, and
\[ \frac{\bar{J}'_\tau}{q'_*(\tau)} < \bar{J}_t + \epsilon.\]

Continuing using a technique similar to that in \cite[Section
11]{stol:ldexp}, we can show that for an arbitrary $S > 0$ and
sufficiently large $n$, we can find intervals $[t_1^{(n)},t_2^{(n)}]$
and $[t_1,t_2]$ such that $\tau \in [t_1^{(n)},t_2^{(n)}] \subset
[t_1,t_2]$, and with

\begin{align}
  &t_2^{(n)} - t_1^{(n)} = \frac{S}{\sqrt{n}}
  \sqrt{\bar{q}^{(n)}\left(t_1^{(n)}\right)}, \nonumber \\
  &\tilde{q}^{(n)}_*(t_2^{(n)}) - \tilde{q}^{(n)}_*(t_1^{(n)}) > 0, \label{eqn:maxqincrease} \\
  &\frac{\bar{J}^{(n)}_{t_2^{(n)}} - \bar{J}^{(n)}_{t_1^{(n)}} }{
    \tilde{q}^{(n)}_*(t_2^{(n)}) - \tilde{q}^{(n)}_*(t_1^{(n)}) } <
  \bar{J}_t + 3\epsilon. \label{eqn:ratio}
\end{align}

We can choose a subsequence of $\{n\}$ above so that, for some $\tau_1
\in [t_1,t_2]$, we have the left endpoints $t_1^{(n)} \to \tau_1$ (so
that $t_2^{(n)} \to \tau_1$ as well). Then, let us choose a further
subsequence such that 
\[\left(\di q^{(n)}, \di q_*^{(n)}, \di f^{(n)},
  \di \hat{f}^{(n)}, (\di c_\alpha^{(n)})_\alpha, (\di
  g_r^{\alpha,(n)})_{\alpha r}, (\di
  \hat{g}_{ri}^{\alpha,(n)})_{\alpha r,i}, \di m^{(n)}\right)\]
converges to an LFSP $\left(\di q, \di q_*, \di f, \di \hat{f}, (\di
  g_r^{\alpha})_{\alpha r}, (\di \hat{g}_{ri}^{\alpha})_{\alpha r,i},
  \di m\right)$ on the {\em local} time interval $[0,S]$.

We claim that there must exist $\epsilon_1 > 0$ such that 
\begin{equation*}
  \hat{J}_S(\di c, \phi) - \hat{J}_0(\di c, \phi) \geq \epsilon_1 S 
\end{equation*}
(recall that $\phi = \frac{d \di g}{d \di c}$ as defined in Lemma
\ref{lem:lfspdynamics}). If not, then $\hat{J}_S(\di c, \phi) =
\hat{J}_0(\di c, \phi)$, which means that all observed channel state
distributions over subsets are {\em exactly} typical. Since, by
hypothesis, the arrival rate vector $\lambda$ lies in the interior of
the throughput region, this contradicts the fact that the longest
queue in the system does not decrease (\ref{eqn:maxqincrease}).

The inequality above, together with the lower bound
(\ref{eqn:convexity2}) and the relation (\ref{eqn:ratio}), gives us that
there exists $\epsilon_2 > 0$ such that
\begin{equation}
  \di q_*(S) - \di q_*(0) \geq \epsilon_2 S. \label{eqn:lfspmaxqinc}
\end{equation}
We thus arrive at the inequality
\begin{equation}
  \frac{\hat{J}_S(\di c, \phi) - \hat{J}_0(\di c, \phi)}{\di q_*(S) - \di q_*(0)} \leq \bar{J}_t + 3\epsilon. \label{eqn:lfspcost}
\end{equation}
In other words, we are able to approximate the cost of FSPs
arbitrarily well with the ``unit cost'' of raising $\di q_*$ in
suitably constructed LFSPs. \\

\subsubsection{A Relaxed Lower Bound on Rate Function in Terms of LFSP
  Costs}
We use the techniques of the previous section to further lower-bound
the queue overflow exponent of the Max-Exp rule. For a general LFSP,
we introduce the following ``potential function'' of its queue state:
\[\Psi(\di q) \bydef \max_{\alpha \in \mathcal{O}} \Psi_\alpha(\di q) \equiv \max_{\alpha \in \mathcal{O}} \sum_{i \in \alpha} e^{\di q_i + b_i}, \]
together with its logarithm
\[ \Phi(\di q) \bydef \log \Psi(\di q) = \max_{\alpha} \log \Psi_\alpha(\di q).\]

\emph{Fact:} The function $\Phi(\di q)$ uniformly approximates $\di
q_* \equiv ||\di q||_\infty$, in the sense that $||\Phi(\di q) - \di
q_*|| \leq \Delta$ for some fixed $\Delta > 0$. \\

Now, consider an FSP feasible for the infimum\footnote{If the infimum
  is not attainable, it suffices to consider an FSP $\epsilon'$-close
  to the infimum, with $\epsilon' > 0$.} (\ref{eqn:stationarylbref})
in Theorem \ref{thm:stationarylbref}. By combining the above fact with
the conclusions of the previous section (i.e. properties
(\ref{eqn:lfspmaxqinc}) and (\ref{eqn:lfspcost})), we have that for an
arbitrarily small $\epsilon > 0$, an LFSP can be constructed on
$[0,S]$, with $S > 0$ suitably large, so that the following properties
hold with $\epsilon_2 > 0$:
\begin{align}
  \Phi(\di q(S)) - \Phi(\di q(0)) &\geq (\epsilon_2/2) S, \label{eqn:lfspcost3}\\
  \frac{\hat{J}_S(\di c, \phi) - \hat{J}_0(\di c, \phi)}{\Phi(\di q(S)) -
    \Phi(\di q(0))} &\leq \bar{J}_t + 2\epsilon. \label{eqn:lfspcost4}
\end{align}
In the sequel, we will concentrate on the LHS of (\ref{eqn:lfspcost4})
-- modulo an arbitrarily small $\epsilon > 0$, it is a lower bound on
the original FSP cost $\bar{J}_t$. We have
\begin{align} 
\frac{\hat{J}_S(\di c, \phi) - \hat{J}_0(\di c, \phi)}{\Phi(\di q(S)) -
    \Phi(\di q(0))} &= \frac{\int_0^S \frac{d}{ds} \hat{J}_s(\di c, \phi) ds }{\int_0^S \frac{d}{ds} \Phi(\di q(s))ds} \nonumber \\
  &\geq \inf_{s \in [0,S]} \frac{\frac{d}{ds} \hat{J}_s(\di c, \phi)}{\frac{d}{ds} \Phi(\di q(s))} \nonumber \\
  &= \inf_{s \in [0,S]} \frac{\sum_{\alpha}
  \di \dot{c}_\alpha(s) \; \Lambda_\alpha^* \left( \phi_\alpha(s) \right)}{\frac{d}{ds} \Phi(\di q(s))} \quad \mbox{(Lemma \ref{lem:lfspdynamics})}. \nonumber 
\end{align}
As a consequence of the above inequality\footnote{We have abused
  notation to indicate that the infimum above is, in fact, over the
  (Lebesgue-a.e.) \emph{regular} points $s \in [0,S]$.}, we can record
the following result:
\begin{proposition}
  \label{prop:lfsplb}
  If $\mathbb{P}$ denotes the stationary measure induced by the
  Max-Exp policy, then
  \begin{align}
    -&\limsup_{n \to \infty} \frac{1}{n} \log \mathbb{P}\left[||q^{(n)}(0)||_\infty \geq 1\right] \geq \; \inf_{s \in [0,S]} \frac{\sum_{\alpha}
  \di \dot{c}_\alpha(s) \; \Lambda_\alpha^* \left( \phi_\alpha(s) \right)}{\frac{d}{ds} \Phi(\di q(s))} \label{eqn:lblfsp},
  \end{align}  
  for any valid Local Fluid Sample Path (LFSP) as specified by
  (\ref{eqn:lfspdef}).
\end{proposition}
Letting $J_*$ denote the infimum on the RHS of (\ref{eqn:lblfsp})
over \emph{all} valid LFSPs, a further lower bound on the buffer
overflow exponent of Max-Exp is thus $J_*$.  \\
  
\subsubsection{Connecting the relaxed Lower Bound to the uniform Upper
  Bound}
The crucial final step in establishing the large-deviations optimality
of the Max-Exp algorithm is to show that the lower bound on its decay
exponent $J_*$ is, in fact, a uniform upper bound on the decay
exponent of any stabilizing scheduling policy, on the lines of Theorem
\ref{thm:ubanysubset}. The proof may be found in Appendix
\ref{app:mexpopt}, and uses the disjointness of the collection of
observable subsets $\mathcal{O}$ in a key way. 

\begin{theorem}[Optimality of Max-Exp]
  \label{thm:mexpopt}
  Let $\pi$ be any stabilizing scheduling policy (i.e., a stabilizing
  policy that bases its decision on the current (or any finite past
  history of) system state) for arrival rates $\lambda = (\lambda_1,
  \ldots, \lambda_n)$, and let $\mathbb{P}^\pi$ be the associated
  stationary measure. Then,
  \[ -\liminf_{n \to \infty} \frac{1}{n} \log
  \mathbb{P}^\pi\left[||q^{(n)}(0)||_\infty \geq 1\right] \leq J_*,\]
  i.e., Max-Exp has the optimal large-deviations exponent (equal to
  $J_*$) over all stabilizing scheduling policies with subset-based
  partial channel state information.
\end{theorem}

\section{Conclusion}
For scheduling with only partial wireless Channel State Information
(CSI), we developed the Max-Exp and Max-Queue scheduling algorithms
yielding optimal queue overflow tails. This work shows that
structurally simple scheduling algorithms which use partial CSI can
guarantee high performance. Moreover, to control queue backlogs in
such cases, no additional statistical or extraneous information is
explicitly required by the scheduling algorithms.

We hope that this work lays the keystone for further investigations of
the performance of wireless scheduling under different types of
partial information structures. Future directions for research include
studying scheduling with information from general user subsets,
temporally varying constraints on available CSI, and performance under
delayed CSI with time-correlated channels.



\appendix

\section{Proof of Proposition \ref{prop:lbT}}
\label{app:lbT}
In the $n$th system, consider the joint channel states for the first
$nT$ time slots, i.e., $\left(R^{(n)}(1),\right.$ $\left.R^{(n)}(k),
  \ldots, R^{(n)}(nT)\right)$, with each $R^{(n)}(k) \in \mathcal{R}^N
\subset \mathbb{R}^N$. Since our sampling/scheduling rule is
deterministic, the exact time slots in $\{1,\ldots,nT\}$ at which user
$i$ is sampled depend entirely on these joint channel states. To avoid
heavy notation, we will suppress the superscript $(n)$ as all
quantities we deal with refer to the $n$th queueing system. Let $V =
(V_1, \ldots, V_N)$ be the (random) \emph{sampled trace} for the
system upto time $nT$. By this, we mean that each $V_i$ is a vector
with elements from $\mathcal{R}$ that represents all the successively
observed/sampled rates for user $i$, i.e. $V_i = \left(R_i(K_{i_1}),
  R_i(K_{i_2}), \ldots \right)$ where user $i$ is chosen precisely at
time slots $K_{i_1}, K_{i_2}, \ldots $ In other words, $V_i$ is the
ordered row of channel state values sampled by the scheduling policy,
so the sum of the lengths of the $V_i$ is exactly $nT$. In the sequel,
we frequently identify each $V_i$ bijectively with its corresponding
partial sums process $W_i \equiv W(V_i)$.
  
We have the following lemma, due to the crucial fact that {\em for any
  deterministic sampling rule, the sampled trace uniquely specifies at
  what times each user was sampled and its sampled channel states at
  those instants}. By a \emph{valid} sampled trace, we mean a (finite)
sampled trace occurring with nonzero probability. For a valid sampled
trace $w$ in the $n$-th system, let $E(w)$ be the set of all extended
combinations of $w$, i.e. the set of all $(e_1, \ldots, e_N)$ where
each $e_i$ is a vector in $\mathcal{R}^{nT}$ such that $w_i$ is a
prefix of $e_i$.

\begin{lemma}
\label{lem:equivprob}
Let $Z_{ij}$, $i = 1,\ldots, N$, $j = 1, 2, \ldots, nT$ be independent
random variables with $Z_{ij} \sim R_i(0)$ for all $i$ and $j$. Let
$\hat{\mathbb{P}}^{(nT)}$ be the probability measure induced by
$(Z_{ij})_{i,j}$. If $w$ is a valid trace in the $n$-th system, then
for any $n q_0 \in (\mathbb{Z}^+)^N$,
\[\mathbb{P}_{q_0}^{n,T}[W^{(n)} = w] = \hat{\mathbb{P}}^{(nT)}\left[E(w) \right]. \]
\end{lemma}

\begin{proof}
  Let $w = (w_1,\ldots,w_N)$ with $\sum_{i=1}^N |w_i| = nT$, and let
  $v = (v_1,\ldots,v_N)$ be the corresponding sampled trace for $w$,
  i.e., each $w_i$ is the vector of partial sums for the vector
  $v_i$. Associated to $w$ and $v$ are the time slots $k_{i_1},
  k_{i_2}, \ldots$ when user $i$ is sampled, for all $i$. Furthermore,
  a key fact is that all the time slots $k_{i_1}, k_{i_2}, \ldots$
  when user $i$ is sampled, for all $i$, are completely specified by
  $v$ due to the sampling rule being nonrandom.

  Recall, from our notation, that the random variable $S(k)$ records
  which user is sampled at time slot $k$. We have
  \begin{align*}
    \mathbb{P}_{q_0}^{n,T}[W^{(n)} = w] &= \mathbb{P}_{q_0}^{n,T}[V^{(n)} = v] \\
    &\stackrel{(a)}{=} \mathbb{P}_{q_0}^{n,T}[V^{(n)} = v, \forall i \; S(k_{i_1}) = i, S(k_{i_2}) = i, \ldots] \\
    &= \mathbb{P}_{q_0}^{n,T}\left[\forall i \; R_i(k_{i_1}) = v_{i1}, R_i(k_{i_2}) = v_{i2}, \ldots \right] \\
    &\stackrel{(b)}{=} \prod_{i,j} \mathbb{P}_{q_0}^{n,T}\left[ R_i(k_{i_j}) = v_{ij} \right] \\
    &\stackrel{(c)}{=} \hat{\mathbb{P}}^{(nT)}\left[E(w) \right],
  \end{align*}
  which completes the proof. Here, $(a)$ is by using the key fact in
  the preceding paragraph; $(b)$ is because channel states are
  independent across time and the fact that the ${k_{i_j}}_{i,j}$ are
  all distinct and partition $\{1,2,\ldots,nT\}$; $(c)$ is due to
  exchangeability of the (independent across time) channel state
  process.
\end{proof}
Proceeding with the proof of the proposition, we have
\begin{align*}
\mathbb{P}_{q_0}^{n,T}\left[q^{(n)} \in \Gamma\right] &\leq \mathbb{P}_{q_0}^{n,T}\left[w^{(n)} \in \Gamma^{(n)} \right], 
\end{align*}
where $\Gamma^{(n)}$ is the set of all valid sampled traces $w^{(n)}$
that result in queue length paths $q^{(n)} \in \Gamma$ under the
scheduling algorithm. For an arbitrary integer $\hat{n}$, we
can write
\begin{align*}
\mathbb{P}_{q_0}^{n,T}\left[w^{(n)} \in \Gamma^{(n)} \right] &= \sum_{w \in \Gamma^{(n)}} \mathbb{P}_{q_0}^{n,T} \left[w^{(n)} = w \right] \\
&= \sum_{w \in \Gamma^{(n)}} \hat{\mathbb{P}}^{(nT)} \left[E(w)\right] \quad \mbox{(by Lemma \ref{lem:equivprob})}\\
&=  \hat{\mathbb{P}}^{(nT)} \left[\bigcup_{w \in \Gamma^{(n)}} E(w)\right] \quad \mbox{(unique prefixes $\Rightarrow$ disjointness)}\\
&\leq \hat{\mathbb{P}}^{(nT)} \left[\bigcup_{n'=\hat{n}}^\infty E\left(\Gamma^{(n')}\right) \right].
\end{align*}
\begin{align}
\therefore \;  -\limsup_{n \to \infty} n^{-1} \log  \mathbb{P}_{q_0}^{n,T}\left[w^{(n)} \in \Gamma^{(n)} \right] &\geq -\limsup_{n \to \infty} n^{-1} \log \hat{\mathbb{P}}^{(nT)} \left[\bigcup_{n'=\hat{n}}^\infty E\left(\Gamma^{(n')}\right) \right] \nonumber \\
&\geq \inf \left\{ \int_0^T \sum_{i=1}^N \Lambda_i^*(\dot{w}_i(z)) dz: w \in \overline{\bigcup_{n'=\hat{n}}^\infty E\left(\Gamma^{(n')}\right)} \right\}  \nonumber \\
&\quad \mbox{(by Mogulskii's theorem \cite{dembo})} \nonumber \\
\Rightarrow -\limsup_{n \to \infty} n^{-1} \log  \mathbb{P}_{q_0}^{n,T}\left[w^{(n)} \in \Gamma^{(n)} \right] &\geq \lim_{\hat{n} \to \infty} \inf \left\{ \int_0^T \sum_{i=1}^N \Lambda_i^*(\dot{w}_i(z)) dz: w \in \overline{\bigcup_{n'=\hat{n}}^\infty E\left(\Gamma^{(n')}\right)} \right\}. \label{eqn:lbw}
\end{align}
Let the right hand side of (\ref{eqn:lbw}) be denoted by $\zeta$. For
every $\hat{n} = 1, 2, \ldots$, we can choose $w_{\hat{n}}$ such that
\begin{align*}
  &w_{\hat{n}} \in \overline{\bigcup_{n'=\hat{n}}^\infty E\left(\Gamma^{(n')}\right)}, \quad \mbox{and}\\  
  &\lim_{\hat{n} \to \infty} \int_0^T \sum_{i=1}^N \Lambda_i^*(\dot{w}_{\hat{n},i}(z)) dz  = \zeta.
\end{align*} 
Since the $w_{\hat{n}}$ are all uniformly Lipschitz continuous and
bounded, by the Arzel\`{a}-Ascoli theorem, the sequence
$\left(w_{\hat{n}}\right)_{\hat{n}}$ contains a subsequence converging
uniformly over the time interval $[0,T]$. Without loss of generality,
let the subsequence be $\{\hat{n}\}$ itself, and let $\lim_{\hat{n}
  \to \infty} w_{\hat{n}} = w$. The map $f \mapsto \int_0^T
\sum_{i=1}^N \Lambda_i^*(\dot{f}(z)) dz$ is lower-semicontinuous
\cite{dembo}, thus
\[ \int_0^T \sum_{i=1}^N \Lambda_i^*(\dot{w}(z)) dz \leq \lim_{\hat{n}
  \to \infty} \int_0^T \sum_{i=1}^N
\Lambda_i^*(\dot{w}_{\hat{n},i}(z)) dz = \zeta.\] We can pick, for
each $\hat{n}$, an $\hat{m}_{\hat{n}}$ and a $w_{\hat{m}_{\hat{n}}}
\in \bigcup_{n'=\hat{n}}^\infty E\left(\Gamma^{(n')}\right)$ such that
$||w_{\hat{m}_{\hat{n}}} - w_{\hat{n}}||_\infty < 1/\hat{n}$. Since
$w_{\hat{m}_{\hat{n}}} \in \bigcup_{n'=\hat{n}}^\infty
E\left(\Gamma^{(n')}\right)$, let $w_{\hat{m}_{\hat{n}}} \in
E\left(\Gamma^{(\hat{m}_{\hat{n}}')}\right)$ for some
$\hat{m}_{\hat{n}}' \geq \hat{n}$. It follows that there exists a
corresponding valid queue length path $q_{\hat{m}_{\hat{n}}}$ such
that $w_{\hat{m}_{\hat{n}}}$ induces $q_{\hat{m}_{\hat{n}}}$, and
moreover, $q_{\hat{m}_{\hat{n}}} \in \Gamma$. We can pick a
subsequence of $\{\hat{m}_{\hat{n}}\}_{\hat{n}}$ (let it be
$\hat{m}_{\hat{n}}$ without loss of generality) along which the
sequence $q_{\hat{m}_{\hat{n}}}$ converges to a $q \in
\overline{\Gamma} = \Gamma$. We now have $\lim_{\hat{n} \to \infty}
w_{\hat{m}_{\hat{n}}} = w$ and $\lim_{\hat{n} \to \infty}
q_{\hat{m}_{\hat{n}}} = q$, thus $(q,w)$ is a valid fluid sample path
with $q \in {\Gamma}$. This yields
\begin{align*}
-\limsup_{n \to \infty} n^{-1} \log  &\mathbb{P}_{q_0}^{n,T}\left[w^{(n)} \in \Gamma^{(n)} \right] \geq \\
\inf &\left\{ \int_0^T \sum_{i=1}^N \Lambda_i^*(\dot{w}_i(z)) dz: (w,q) \mbox{ an FSP}, q \in {\Gamma} \right\}
\end{align*}
\begin{align}
\Rightarrow -\limsup_{n \to \infty} \frac{1}{n} \log &\mathbb{P}_{q_0}^{n,T}\left[q^{(n)} \in \Gamma\right] \geq \nonumber \\
\inf &\left\{ \int_0^T \sum_{i=1}^N \Lambda_i^*(\dot{w}_i(z)) dz: (w,q) \mbox{ an FSP}, q \in {\Gamma} \right\}. \label{eqn:infwq}
\end{align}

\emph{Note:} By $(w,q)$ being an FSP, in addition to there existing
prelimit sequences $w^{(n)} \to w$ and $q^{(n)} \to q$ (uniformly over
$[0,T]$), we mean that there exist points $z_i \in [0,T]$ for all $i =
1,\ldots,N$ such that $z_i^{(n)} \to t_i$ as $n \to \infty$, where
$z_i^{(n)}$ is the (scaled by $1/n$) index in the sampled trace
$w_i^{(n)}$ beyond which user $i$ is never sampled (i.e. it is the
last index at which user $i$ is sampled by the scheduling algorithm if
user $i$'s samples are stacked successively and contiguously).

From the observation preceding Lemma \ref{lem:equivprob}, each sampled
trace $w^{(n)}$ completely specifies the exact instants at which each
user was scheduled/sampled and the channel states observed at those
instants, i.e. $w^{(n)}$ completely specifies the pair $(m^{(n)},
c^{(n)})$ in $[0,T]$. The next lemma relates the large deviations
``costs'' of the fluid limits of sampled traces to those of the fluid
limits of their associated $(m^{(n)}, c^{(n)})$ processes.

\begin{lemma}
\label{lem:wtomc}
Let $(w,q)$ be a fluid sample path with $w^{(n)} \to w$ and $q^{(n)}
\to q$. For each integer $n \geq 1$, let $(m^{(n)}, c^{(n)})$ be the
scaled sampled rate and selection processes, which are completely
specified by $w^{(n)}$. Then, for every subsequential limit $(m,c)$ of
$(m^{(n)}, c^{(n)})_n$ (in the $||\cdot||_\infty$ topology on
$[0,T]$),
\begin{align*}
\sum_{i=1}^N \int_{0}^{z_i} \Lambda_i^*(\dot{w}_i(z)) dz =
\int_{0}^T \left[\sum_{i=1}^N \dot{c_i}(t) \Lambda_i^*
  \left(\frac{\dot{m_i}(t)}{\dot{c_i}(t)}\right)\right] dt.
\end{align*}
\end{lemma}

\begin{proof}
Assume without loss of generality that $m^{(n)} \to m$ and $c^{(n)}
\to c$ uniformly in $[0,T]$. Let $0 \leq t_1 \leq t_2 \leq T$. For
all $n$, by the definition of the sampled traces $w_i^{(n)}$, we have
\begin{align}
w_i^{(n)}\left(c_i^{(n)}(t_2)\right) - w_i^{(n)}\left(c_i^{(n)}(t_1)\right) &= m_i^{(n)}(t_2)
- m_i^{(n)}(t_1) + O(1/n). \label{eqn:prelimitmeas}
\end{align}
By the (uniform) convergence hypotheses, for $j \in \{1,2\}$,
$c_i^{(n)}(t_j) \to c_i(t_j)$, thus
$w_i^{(n)}\left(c_i^{(n)}(t_j)\right) \to w_i(c_i(t_j))$. Letting $n
\to \infty$ in (\ref{eqn:prelimitmeas}), 
\begin{align}
w_i\left(c_i(t_2)\right) - w_i\left(c_i(t_1)\right) &= m_i(t_2)
- m_i(t_1). \label{eqn:limitmeas}
\end{align}
Since $c_i$ and $m_i$ are nondecreasing Lipschitz-continuous
functions, they induces Stieltjes measures $dc_i$ and $dm_i$
respectively on $[0,T]$ with $dm_i << dc_i$. In a similar fashion,
$w_i$ induces a Stieltjes measure $dw_i << dz$ on $[0,z_i]$ where
$dz$ denotes Lebesgue measure. Let $dw_i/dz$ be the Radon-Nikodym
derivative of $dw_i$ with respect to Lebesgue measure, and consider
\begin{align*}
\int_{t_1}^{t_2} \frac{dw_i}{dz} \circ c_i (t) dc_i(t) &= \int_{c_i(t_1)}^{c_i(t_2)} \frac{dw_i}{dz} (z) \left( dc_i \circ c^{-1}\right) \quad \mbox{(change of variables formula)} \\
&= \int_{c_i(t_1)}^{c_i(t_2)} \frac{dw_i}{dz} (z) dz \quad \mbox{($dc_i \circ c^{-1} \equiv$ Lebesgue$[0,z_i]$)} \\
&= w_i\left(c_i(t_2)\right) - w_i\left(c_i(t_1)\right) \\
&= m_i(t_2) - m_i(t_1) \quad \mbox{(thanks to (\ref{eqn:limitmeas}))} \\
&= \int_{t_1}^{t_2} dm_i(t) \\
\Rightarrow \quad \frac{dw_i}{dz} \circ c_i(\cdot)  &= \frac{dm_i}{dc_i}(\cdot) \quad dc_i\mbox{-a.e. on } [0,T]. 
\end{align*}
With this, we can finally compute
\begin{align*}
&\int_0^T \dot{c_i}(t) \Lambda_i^*\left(\frac{\dot{m_i}(t)}{\dot{c_i}(t)}\right) = \int_0^T \left(\Lambda^* \circ \frac{dm_i}{dc_i}\right) (t) dc_i(t) \\
&= \int_0^T \left(\Lambda^* \circ \frac{dw_i}{dz} \circ c_i \right) (t) dc_i(t) = \int_{c_i(0)}^{c_i(T)} \left(\Lambda^* \circ \frac{dw_i}{dz}\right) (z) \left( dc_i \circ c^{-1}\right) \\
&= \int_{0}^{z_i} \left(\Lambda^* \circ \frac{dw_i}{dz}\right) (z) dz = \int_{0}^{z_i} \Lambda^*\left(\dot{w}_i(z)\right) dz.
\end{align*}
This proves the lemma.
\end{proof}
Applying the result of Lemma \ref{lem:wtomc} to (\ref{eqn:infwq})
concludes the proof of Proposition \ref{prop:lbT}.

\section{Proof of Proposition \ref{prop:stationarylb}}
\label{app:stationarylb}
With reference to the proof of a similar result \cite[Theorem
8.4]{stol:ldexp}, we can establish the following properties in a
completely analogous fashion to complete the proof of the proposition
(the proofs are omitted to avoid repetition):

\begin{lemma}
Let $\delta > 0$ and $c > 0$ be given, and let the stopping time
$\beta^{(n)} \bydef \inf \{t \geq 0: ||q^{(n)}(t)||_\infty \leq \delta
\}$. Then, there exists $\Delta > 0$ such that
\[\limsup_{n \to \infty} \sup_{y: ||q(y)||_\infty \leq c} \mathbb{E}_y \beta^{(n)} \leq \Delta c. \]
\end{lemma}

\begin{lemma}
For fixed constants $c > \delta > 0$ and $T > 0$, let 
\begin{align*}
 K(c,\delta,T) &\bydef \inf_{(m^T,c^T,q^T)} \quad \quad \int_{0}^T \left[\sum_{i=1}^N \dot{c_i}(t) \Lambda_i^* \left(\frac{\dot{m_i}(t)}{\dot{c_i}(t)}\right)\right] dt \\
&\emph{subject to} \quad (m^T,c^T,q^T) \; \emph{an FSP}, \\
&\hspace{2.2cm} ||q(0)||_\infty \leq c, ||q(t)||_\infty \geq \delta \mbox{ for all } 0 \leq t \leq T.
\end{align*}
Then, uniformly over $\delta$, $K(c,\delta,T) \to \infty$ as $T \to
\infty$.
\end{lemma}

\section{Proof of Proposition \ref{prop:ubany}}
\label{app:ubany}
Denote by $\mathcal{C}(\phi_1,\ldots,\phi_N)$ the convex hull of the
points $(0,\ldots,0)$, $(\phi_1,\ldots,0)$, $(0,\phi_2,\ldots,0)$,
$\ldots$ , $(0,\ldots,\phi_N)$. The right-hand side of
(\ref{eqn:liminf}) is trivially $\infty$ if either (a) $\lambda \in
\mathcal{C}(\phi_1,\ldots,\phi_N)$, or (b) any of the $\phi_i$ is
not in the effective domain of its corresponding $\Lambda_i^*$; we
exclude such $\phi_i$ and $\lambda$ in the remainder of the proof.

For each $n = 1, 2, \ldots$, let $t_n \geq 0$ be a nonrandom time, to
be specified later (to avoid complications, we assume $nt_n$ is an
integer). Consider
\begin{align*}
\mathbb{P}^\pi\left[||q^{(n)}(0)||_\infty \geq 1 \right] &= \mathbb{P}^\pi\left[||q^{(n)}(t_n)||_\infty \geq 1 \right] \\
&\geq \mathbb{P}^\pi\left[||q^{(n)}(t_n)||_\infty \geq 1 \mid ||q^{(n)}(0)||_\infty = 0 \right] \mathbb{P}^\pi\left[||q^{(n)}(0)||_\infty = 0 \right] \\
&= \pi((0,0,\ldots,0)) \; \mathbb{P}^\pi_0\left[||q^{(n)}(t_n)||_\infty \geq 1 \right].
\end{align*}
Here, $\pi(\cdot)$ is used to denote the stationary distribution that
the policy $\pi$ induces, and $\mathbb{P}^\pi_0$ represents the
stationary distribution conditioned on the starting state being the
origin (all zeroes). 

The non-negativity of queues forces the relation $U^{(n)}(k) \bydef
\lambda n - M^{(n)}(k) \leq Q^{(n)}(k)$, where $U^{(n)}(k) =
\sum_{l=1}^k (\lambda - R^{(n)}(l)\delta_{S(l)})$ represents the
``unreflected queue lengths'' in the $n$-th queueing system at time
$k$. By suitably rescaling in time and space, we can continue this
chain of inequalities as
\begin{align}
\label{eqn:chaincontinue}
\mathbb{P}^\pi\left[||q^{(n)}(0)||_\infty \geq 1 \right] &\geq \pi((0,0,\ldots,0)) \; \mathbb{P}^\pi_0\left[\max_i u^{(n)}_i(t_n) \geq 1 \right].
\end{align}
For each $i = 1,\ldots, N$, since $\phi_i$ is in the effective domain
of its Cram\'{e}r rate function $\Lambda_i^*$, it follows that there
exists $\eta_i' \in \mathbb{R}$ such that $\Lambda^*_i(\phi_i) =
\eta_i' \phi_i - \Lambda_i(\eta_i')$. Define for each $i$ an
exponentially tilted measure $\hat{\mathbb{P}}_i$ (with respect to the
marginal measure $\mathbb{P}_i$ of the $i$-th channel state $R_i(0)$)
on $\mathbb{R}$ as follows:
\[ \hat{\mathbb{P}}_i (dx) \bydef \exp[{\eta_i' x  - \Lambda_i(\eta_i')}] \; \mathbb{P}_i(dx) = \exp[{\eta_i' (x - \phi_i) + \Lambda^*_i(\phi_i)}] \; \mathbb{P}_i(dx). \]
A standard computation under the tilted measure yields
$\hat{\mathbb{E}}_i[R_i(0)] = \phi_i$. As with the approach followed
in \cite{shimkin:extremal}, let $\hat{\mathbb{P}}^\pi_0$ be the
measure defined similarly to $\mathbb{P}^\pi_0$ except that the
twisted measures $\{\hat{\mathbb{P}}_i\}$ replace $\{{\mathbb{P}}_i\}$
as the conditional marginal distributions of the sampled channel
states/rates, with $\{\hat{\mathbb{E}}_i\}$ being the corresponding
expectations. 

Let us define 
\begin{align*}
t_{\min}^{-1} &\bydef \max_i \lambda_i,  \\
t_{\max}^{-1} &\bydef \min_{\mu \in \mathcal{C}(\phi_1,\ldots,\phi_N)} \max_i (\lambda_i - \mu_i). 
\end{align*}
Since by hypothesis the arrival rate $\lambda$ is outside the closed
set $\mathcal{C}(\phi_1,\ldots,\phi_N)$, it follows that $0 <
t_{\min} \leq t_{\max} < \infty$. The times $t_{\min}$ and $t_{\max}$
represent the earliest and latest time that the maximum queue length
can take to overflow to level 1 in a system of queues with ``fluid''
inputs at rates $\lambda_i$ that can be drained with instantaneous
rates in the convex hull $\mathcal{C}(\phi_1,\ldots,\phi_N)$. 

The remainder of the proof is organized into four steps:
\begin{enumerate}
\item \label{enum:step1} Showing that for $n$ large enough, under the
  twisted measure $\hat{\mathbb{P}}$, the service $m_i^{(n)}(t)$
  provided to the queue $i$ is approximated with high probability by
  $\phi_i c_i(t)$, i.e. we can treat the channel as being
  deterministic with a service rate of $\phi_i$,
\item Under the conditions of the previous step, overflow of the
  unreflected max-queue $d^{(n)}(\cdot)$ is inevitable by time roughly
  $t_{\max}$, so with a significant probability the first hitting time
  of $d^{(n)}(\cdot)$ to level 1 is at most $t_{\max}$. Thus, we can
  find a time not exceeding $t_{\max}$ at which overflow occurs with a
  significant probability (i.e. not decaying to 0 exponentially in
  $n$)
\item Overflow occurring at the time in the previous step, under the
  conditions of step \ref{enum:step1}, forces the scheduling ``choice
  fractions'' $c^{(n)}(t)/t$ to be ``consistent'' with overflow of
  $d^{(n)}(\cdot)$ occurring at that time
\item Using all the steps to develop the right-hand side of
  (\ref{eqn:chaincontinue}) and derive the stated result.
\end{enumerate}

\subsubsection{Step 1 of 4}
Let us record the following definition. For each $i = 1, \ldots, N$,
we can write
\begin{align*}
  m_i^{(n)}(t) \equiv m_i(t) = \frac{1}{n} M_i(nt) &= \frac{1}{n} \sum_{l=0}^{nt} R_i(l) X_i(l) = \frac{\overline{M}_i(nt)}{n}  + \frac{\phi_i}{n} \sum_{l=0}^{nt}  X_i(l) \\
  &= \frac{\overline{M}_i(nt)}{n}  + \frac{\phi_i}{n} C_i(nt), \\
  &= \overline{m}_i(t) + \phi_ic_i(t),
\end{align*}
where $X_i(l)$ is the indicator of the event that user $i$ was
scheduled at time slot $l$, and $\overline{M}_i(k) \bydef
\sum_{l=0}^{k} (R_i(l) - \phi_i) X_i(l)$ is the (unscaled)
``centered'' service provided to queue $i$ upto time slot $k$.
\begin{lemma}
  \label{lem:approx}
  Let times $t_1$ and $t_2$, such that $0 < t_1 \leq t_2$, and $\delta
  > 0$ be fixed. Then,
  \[ \lim_{n \to \infty} \hat{\mathbb{P}}_0^{\pi} \left[ |\overline{m}^{(n)}_i(t)| < \delta t \quad \forall t \in \left[t_1, t_2 \right]\right] = 1. \]
\end{lemma}

\begin{proof}
  Observe that for each $i$, $\{\overline{M}_i(k)\}_k$ is a martingale
  (with respect to the measure $\hat{\mathbb{P}^\pi_0}$) null at 0 and
  with differences bounded by $D \bydef (R_{\max} + \max_i \phi_i)$,
  where $R_{\max}$ is the maximum channel rate across all channels in
  the system. An application of the Azuma-Hoeffding martingale
  inequality \cite{mitzenmacherupfal05:probcomp} thus gives
  \begin{equation}
    \label{eqn:azuma}
    \hat{\mathbb{P}}_0^{\pi} \left[\left|\frac{\overline{M}_i(k)}{k}\right| \geq \gamma\right] \leq 2e^{-\frac{k \gamma^2}{2 D^2}}
  \end{equation}
  for all $k = 1, 2, \ldots$. Hence, a union bound gives 
  \begin{align*}
    1 - &\hat{\mathbb{P}}_0^{\pi} \left[ |\overline{m}_i(t)| < \delta t \quad \forall t \in \left[t_1, t_2\right]\right] = \hat{\mathbb{P}}_0^{\pi} \left[ \exists t \in \left[t_1, t_2\right] \; |\overline{m}_i(t)| \geq \delta t\right] \\
    &\leq \sum_{k = nt_1}^{nt_2} \hat{\mathbb{P}}_0^{\pi} \left[ \left|\frac{\overline{M}_i(k)}{k}\right| \geq \delta \right] \\
    &\leq \sum_{k = nt_1}^{nt_2} 2e^{-\frac{k \delta^2}{2 D^2}} \\
    &\leq 2n(t_2-t_1) e^{-\frac{nt_1 \delta^2}{2 D^2}} \; \stackrel{n \to \infty}{\longrightarrow} 0 \quad \mbox{($\because$ $t_1 > 0$)}  , 
  \end{align*}
  which is the stated result. 
\end{proof}

\subsubsection{Step 2 of 4}
Let us fix $\delta > 0$ small enough, and let $\epsilon > 0$ be such
that 
\[ (t_{\max} + \epsilon)^{-1} = \min_{\mu \in \mathcal{C}(\phi_1,\ldots,\phi_N)} \max_i (\lambda_i - \delta - \mu_i). \]
Additionally, fix a time $t_0 > 0$ small enough, and let $A \equiv A_n$ denote
the event whose (twisted) probability is estimated in Lemma
\ref{lem:approx}, i.e.
\[ A_n \equiv A_n(\delta) \equiv A_n(\delta,t_0,t_{\max}) \bydef \left\{ |\overline{m}^{(n)}_i(t)| < \delta t \quad \forall t \in \left[t_0, t_{\max}\right] \right\}. \]

Denote the (unreflected and fluid-scaled) maximum queue length process
by $d(\cdot) \equiv d^{(n)}(\cdot) \bydef \max_i u^{(n)}_i(\cdot)$. It
follows that in the event $A_n$, $d(\cdot)$ must overflow (i.e.  hit
level 1) at least once by time $(t_{\max} + \epsilon)$. In other
words, if we let
\[ \tau \equiv \tau_n \bydef \inf \left\{t = 0, \frac{1}{n},
  \frac{2}{n}, \ldots: d(t) \geq 1 \right\}, \]

then $A_n \subseteq \{\tau_n \leq t_{\max} + \epsilon\}$. For each $n
= 1, 2, \ldots$, define the (deterministic) time
\[ t_n \bydef \arg \max_{t = 0, \frac{1}{n}, \ldots, t_{\max} + \epsilon} \hat{\mathbb{P}}_0^{\pi}\left[\tau_n = t \right]  \]
with ties broken in an arbitrary fashion. Observe that $t_n$ does not depend upon $\delta$, $t_0$ or $A_n$.
Also, note that
\begin{align*}
  \hat{\mathbb{P}}_0^{\pi}\left[A_n\right] &\leq \hat{\mathbb{P}}_0^{\pi}\left[\tau_n \leq t_{\max} + \epsilon \right] \leq \sum_{t=0, \frac{1}{n}, \ldots, t_{\max} + \epsilon} \hat{\mathbb{P}}_0^{\pi}\left[\tau_n = t \right] \\
  &\leq n(t_{\max} + \epsilon) \left(\max_{t = 0, \frac{1}{n}, \ldots, t_{\max} + \epsilon} \hat{\mathbb{P}}_0^{\pi}\left[\tau_n = t \right]\right)\\
  &\leq n(t_{\max} + \epsilon) \hat{\mathbb{P}}_0^{\pi}\left[\tau_n = t_n \right]\\
  \Rightarrow \; \hat{\mathbb{P}}_0^{\pi}\left[\tau_n = t_n \right] &\geq \frac{\hat{\mathbb{P}}_0^{\pi}\left[A_n\right]}{n(t_{\max} + \epsilon)}
\end{align*}
Since the rate of change of $d^{(n)}(\cdot)$ is bounded by $D$, we can
write
\begin{align}
  \label{eqn:step2result}
  \Rightarrow \hat{\mathbb{P}}^\pi_0\left[d(t_n) \in \left[1, 1 + \frac{D}{n}\right] \right] &\geq \hat{\mathbb{P}}_0^{\pi}\left[\tau_n = t_n \right] \geq \frac{\hat{\mathbb{P}}_0^{\pi}\left[A_n\right]}{n(t_{\max} + \epsilon)}.
\end{align}

\subsubsection{Step 3 of 4}
This step involves showing that when the queues overflow at time $t_n$
then the scheduling choice fractions $c^{(n)}(t_n)/t_n$ at that time
are very likely to be the ones that cause ``straight-line'' overflow
at time $t_n$ from the all-empty queue state.

Recall that $\delta > 0$ is a sufficiently small number. We denote by
$\Gamma_n$ the set of \emph{$\delta$-compatible} scheduling fractions
for overflow at time $t_n$ as follows:
\[ \Gamma_n \equiv \Gamma_n(\delta, t_n) \bydef \left\{(f_1',\ldots,f_N'): \sum_i f_i' = 1, f_i' \geq 0, \max_i \left(\lambda_i t_n - \phi_i f_i' t_n\right) \in [1-\delta t_n, 1+\delta t_n] \right\}. \]

\begin{lemma}
  \label{lem:compatiblefrac}
  For all $n$ large enough and $\delta > 0$,
  \[\hat{\mathbb{P}}^\pi_0\left[d(t_n) \in \left[1 - \frac{D}{n}, 1 + \frac{D}{n}\right], \frac{c(t_n)}{t_n} \notin \Gamma_n \right] \leq 2N e^{- \frac{nt_0 \delta^2}{8 D^2}} . \]
\end{lemma}

\begin{proof}
  For $n$ sufficiently large,
  \begin{align*}
    d(t_n) \in \left[1 - \frac{D}{n}, 1 + \frac{D}{n}\right] &\Rightarrow d(t_n) \in \left[1 - \frac{\delta}{2} t_n, 1 + \frac{\delta}{2} t_n \right].    
  \end{align*}
  Also, 
  \begin{align*}
    d(t_n) \in \left[1 - \frac{\delta}{2} t_n, 1 + \frac{\delta}{2} t_n \right], \frac{c(t_n)}{t_n} \notin \Gamma_n &\Rightarrow \exists i \; |\overline{m}_i(t_n)| > \frac{\delta}{2} t_n. 
  \end{align*}
  Thus, 
  \begin{align*}
    &\hat{\mathbb{P}}^\pi_0\left[ d(t_n) \in \left[1 - \frac{D}{n}, 1 + \frac{D}{n}\right], \frac{c(t_n)}{t_n} \notin \Gamma_n \right] \leq \hat{\mathbb{P}}^\pi_0\left[\exists i \; |\overline{m}_i(t_n)| > \frac{\delta}{2} t_n \right] \\
    &\leq \sum_i \hat{\mathbb{P}}^\pi_0\left[|\overline{m}_i(t_n)| > \frac{\delta}{2} t_n \right] \leq \sum_i 2e^{- \frac{nt_n \delta^2}{8 D^2}} \quad \mbox{(by the Azuma-Hoeffding inequality (\ref{eqn:azuma}))} \\
    &\leq 2N e^{- \frac{nt_0 \delta^2}{8 D^2}}.
  \end{align*}
\end{proof}

\subsubsection*{Step 4 of 4}
We can now finally develop the right-hand side of
(\ref{eqn:chaincontinue}) using the results from the previous steps:
\begin{align}
  &\mathbb{P}^\pi_0\left[\max_i u^{(n)}_i(t_n) \geq 1 \right]  = \mathbb{P}^\pi_0\left[d^{(n)}(t_n) \geq 1 \right] \nonumber \\
  &\geq \mathbb{P}^\pi_0\left[d^{(n)}(t_n) \geq 1, c^{(n)}(t_n)/t_n \in \Gamma_n \right] \nonumber \\
  &= \mathbb{E}^\pi_0\left[\mathbbm{1}_{\{d(t_n) \geq 1, c(t_n)/t_n \in \Gamma_n\}} \right] \nonumber \\
  &= \hat{\mathbb{E}}^\pi_0\left[\mathbbm{1}_{\{d(t_n) \geq 1, c(t_n)/t_n \in \Gamma_n\}} \prod_{l=0}^{nt_n} \exp\left[-\Lambda^*_{U(l)}(\phi_{U(l)}') - \eta_{U(l)}' \left(R_{U(l)}(l) - \phi_{U(l)}'\right) \right] \right] \nonumber \\
  &= \hat{\mathbb{E}}^\pi_0\left[\mathbbm{1}_{\{d(t_n) \geq 1, c(t_n)/t_n \in \Gamma_n\}} \exp \left[-nt_n \left(\sum_i \frac{c_i(t_n)}{t_n} \Lambda^*_i(\phi_i) \right)  \right] \exp\left[-n w(t_n) \right] \right] \nonumber \\
  &\hspace{20pt} \left(\mbox{with} \; w(t_n) \equiv w^{(n)}(t_n) \bydef \frac{1}{n}W(nt_n) \bydef  \frac{1}{n} \sum_{l=0}^{nt_n} \eta_{U(l)}' \left(R_{U(l)}(l) - \phi_{U(l)}'\right)\right) \nonumber \\
    &= \hat{\mathbb{E}}^\pi_0\left[\mathbbm{1}_{\{d(t_n) \geq 1, c(t_n)/t_n \in \Gamma_n\}} \exp \left[-nt_n \left( \sup_{f' \in \Gamma_n} \sum_i f'_i \Lambda^*_i(\phi_i) \right)  \right] \exp\left[-n w(t_n) \right] \right] \nonumber \\
    &= \exp \left[-nt_n \left( \sup_{f' \in \Gamma_n} \sum_i f'_i \Lambda^*_i(\phi_i) \right)  \right]  \hat{\mathbb{E}}^\pi_0\left[\mathbbm{1}_{\{d(t_n) \geq 1, c(t_n)/t_n \in \Gamma_n\}} e^{-n w(t_n)} \right]. \label{eqn:rhsdevelopment}
\end{align}
The second term in the product above can be bounded from below for any
$\zeta > 0$ as follows:
\begin{align}
  &\hat{\mathbb{E}}^\pi_0\left[\mathbbm{1}_{\{d(t_n) \geq 1, c(t_n)/t_n \in \Gamma_n\}} e^{-n w(t_n)} \right] \geq \hat{\mathbb{E}}^\pi_0\left[\mathbbm{1}_{\{d(t_n) \geq 1, c(t_n)/t_n \in \Gamma_n, |w(t_n)| < \zeta \}} e^{-n \zeta} \right] \nonumber \\
  &=  e^{-n \zeta} \; \hat{\mathbb{P}}^\pi_0\left[d(t_n) \geq 1, \frac{c(t_n)}{t_n} \in \Gamma_n, |w(t_n)| < \zeta \right] \nonumber \\
  &\geq  e^{-n \zeta} \; \hat{\mathbb{P}}^\pi_0\left[d(t_n) \in \left[1, 1+\frac{D}{n} \right], \frac{c(t_n)}{t_n} \in \Gamma_n, |w(t_n)| < \zeta \right]. \label{eqn:lb0}
\end{align}
We have
\begin{align}
  &\hat{\mathbb{P}}^\pi_0\left[d(t_n) \in \left[1, 1+\frac{D}{n} \right], \frac{c(t_n)}{t_n} \in \Gamma_n, |w(t_n)| < \zeta \right] \nonumber \\
  &\geq \hat{\mathbb{P}}^\pi_0\left[d(t_n) \in \left[1, 1+\frac{D}{n} \right], \frac{c(t_n)}{t_n} \in \Gamma_n \right] - \hat{\mathbb{P}}^\pi_0\left[ |w(t_n)| \geq \zeta \right] \nonumber \\
  &\geq \hat{\mathbb{P}}^\pi_0\left[d(t_n) \in \left[1, 1+\frac{D}{n} \right]\right] - \hat{\mathbb{P}}^\pi_0\left[d(t_n) \in \left[1, 1+\frac{D}{n} \right], \frac{c(t_n)}{t_n} \notin \Gamma_n \right] - \hat{\mathbb{P}}^\pi_0\left[ |w(t_n)| \geq \zeta \right]. \label{eqn:problb}
\end{align}
By definition and the properties of the twisted distribution
$\hat{\mathbb{P}}$, it can be seen that $\{W(k)\}_{k=0,1,\ldots}$ is
again a martingale null at 0 and with bounded increments (bounded by,
say, $D_2 \bydef \max_i \eta'_i (R_{\max} + \phi_i)$). Hence, the
Azuma-Hoeffding inequality applied to it yields
\begin{align*}
  \hat{\mathbb{P}}^\pi_0\left[ |w(t_n)| \geq \zeta \right] &\leq 2e^{-\frac{n \zeta^2}{2t_n D_2^2}} \leq 2e^{-\frac{n \zeta^2}{2t_{\max} D_2^2}}.
\end{align*}
Using this and the results of Steps 2 and 3, (\ref{eqn:problb}) becomes
\begin{align*}
  &\hat{\mathbb{P}}^\pi_0\left[d(t_n) \in \left[1, 1+\frac{D}{n} \right], \frac{c(t_n)}{t_n} \in \Gamma_n, |w(t_n)| < \zeta \right] \\
  &\geq \frac{\hat{\mathbb{P}}_0^{\pi}\left[A_n\right]}{n(t_{\max} + \epsilon)} - 2N e^{- \frac{nt_0 \delta^2}{8 D^2}} - 2e^{-\frac{n \zeta^2}{2t_{\max} D_2^2}} \\
  &\geq \frac{1/2}{n(t_{\max} + \epsilon)} - 2N e^{- \frac{nt_0 \delta^2}{8 D^2}} - 2e^{-\frac{n \zeta^2}{2t_{\max} D_2^2}} \quad \mbox{(for $n$ large enough, by Step 1)}. 
\end{align*}
The first term above decays as $n^{-1}$ while the second and third
terms decay exponentially in $n$, thus
\begin{align}
  -\liminf_{n \to \infty} \frac{1}{n} \log \hat{\mathbb{P}}^\pi_0\left[d(t_n) \in \left[1, 1+\frac{D}{n} \right], \frac{c(t_n)}{t_n} \in \Gamma_n, |w(t_n)| < \zeta \right] \leq 0. \label{eqn:lb1}
\end{align}
What remains is to bound the first term in the product in
(\ref{eqn:rhsdevelopment}). By definition, for every $f' \in
\Gamma_n$, we have
\begin{align}
  \max_i (\lambda_i - \phi_i f_i') &\leq \frac{1}{t_n} + \delta \nonumber \\
  \Rightarrow \;   t_n \sup_{f' \in \Gamma_n} \sum_i f'_i \Lambda^*_i(\phi_i)  &\leq \sup_{f' \in \Gamma_n} \frac{\sum_i f'_i \Lambda^*_i(\phi_i)}{\max_i (\lambda_i - \phi_i f_i') - \delta} \nonumber \\
  &\leq \sup_{\substack{\sum_i f_i' = 1\\f_i' \geq 0}} \frac{\sum_i f'_i \Lambda^*_i(\phi_i)}{\max_i (\lambda_i - \phi_i f_i') - \delta}. \label{eqn:lb2}
\end{align}
Applying the conclusions of (\ref{eqn:lb0}), (\ref{eqn:lb1}) and
(\ref{eqn:lb2}) to (\ref{eqn:rhsdevelopment}), we get
\begin{align*}
    -\liminf_{n \to \infty} \frac{1}{n} \log \mathbb{P}^\pi_0\left[\max_i u^{(n)}_i(t_n) \geq 1 \right] &\leq \zeta + \sup_{\substack{\sum_i f_i' = 1\\f_i' \geq 0}} \frac{\sum_i f'_i \Lambda^*_i(\phi_i)}{\max_i (\lambda_i - \phi_i f_i') - \delta}. 
\end{align*}
The arbitrary choice of $\zeta > 0$ and $\delta > 0$ implies that
\begin{align*}
    -\liminf_{n \to \infty} \frac{1}{n} \log \mathbb{P}^\pi_0\left[\max_i u^{(n)}_i(t_n) \geq 1 \right] &\leq \sup_{\substack{\sum_i f_i' = 1\\f_i' \geq 0}} \frac{\sum_i f'_i \Lambda^*_i(\phi_i)}{\max_i (\lambda_i - \phi_i f_i')}. 
\end{align*}
The stationary distribution $\mathbb{P}^{\pi}$ induced by the
(stabilizing) scheduling policy $\pi$ forces $\pi((0,0,\ldots,0)) >
0$, so (\ref{eqn:chaincontinue}) finally implies
\begin{align*}
    -\liminf_{n \to \infty} \frac{1}{n} \log \mathbb{P}^{\pi}\left[ ||q^{(n)}(0)||_\infty \geq 1 \right] \leq \sup_{\substack{\sum_i f_i' = 1\\f_i' \geq 0}} \frac{\sum_i f'_i \Lambda^*_i(\phi_i)}{\max_i (\lambda_i - \phi_i f_i')},
\end{align*}
which completes the proof of the proposition.

\section{Proof of Proposition \ref{prop:mqopt}}
\label{app:mqopt}

  Recall that $J_*$ is the infimum 
  \begin{align}
    \label{eqn:j*def}
      J_* \bydef \inf_{\substack{T,(m^T,c^T,q^T)\\0 \leq t \leq T}}  \frac{\sum_{i=1}^N \dot{c_i}(t) \Lambda_i^* \left(\frac{\dot{m_i}(t)}{\dot{c_i}(t)}\right)}{\frac{d}{dt} ||q(t)||_\infty}
  \end{align}
  over all feasible Fluid Sample Paths at regular points $t$. There is
  nothing to be done if the right hand side above is $\infty$, so we
  exclude this case. We have the following characterization of regular
  points under the Max-Queue scheduling algorithm.
  \begin{lemma}
    \label{lem:regular}
    Under the Max-Queue policy, let $s(t) \bydef \arg
    \max_{i=1,\ldots,N} q_i(t) \subseteq \{1,\ldots,N\}$. If $t$ is a
    regular point, then
    \begin{enumerate}
    \item $c_i'(t) = 0$ $\forall i \notin s(t)$, i.e., the
      non-maximum fluid queues do not receive service,
    \item $\frac{d}{dt}||q(t)||_\infty = \lambda_i - m_i'(t)$ $\forall
      i \in s(t)$, i.e., all the maximum fluid queues grow at the same
      rate.
    \end{enumerate}
  \end{lemma}
  Thus, by Lemma \ref{lem:regular},
  \begin{align} 
    J_* &\geq \inf_{S \subseteq \{1,\ldots,N\}} \frac{\sum_{i \in S}
      c'_i \Lambda^*_i \left(\phi_i\right)}{w'}, \label{eqn:j*lb}
  \end{align}
  for all non-negative $\{c_i'\}_{i \in S}$, $\{\phi_i\}_{i \in S}$
  satisfying $\sum_{i \in S} c'_i = 1$, and $w' = \lambda_i -
  c_i'\phi_i$ $\forall i \in S$. Note that the denominator $w'$ is
  strictly positive if and only if $\lambda \notin
  \mathcal{C}(\phi_1,\ldots,\phi_N)$, and that each $\phi_i$ can be
  restricted to be at most $\mathbb{E}[R_i]$ (since if $\phi_i >
  \mathbb{E}[R_i]$, reducing $\phi_i$ to
  $\mathbb{E}[R_i]$ only gives a lesser fraction above). 

  For a subset $S \subseteq \{1,\ldots,N\}$, let
  \[\mathcal{D}_S \bydef \left\{(\phi_i)_{i \in S}: R_{\min,i} \leq \phi_i \leq \mathbb{E}[R_i], \; \exists c_i' \geq 0 \; \mbox{with} \; \sum_{i \in S} c_i' = 1, \; \forall i,j \in S \; \lambda_i - c_i' \phi_i = \lambda_j - c_j'\phi_j = w' > 0 \right\}. \]
  It follows that for each such tuple $\phi \in \mathcal{D}_S$,
  there is a \emph{unique} corresponding tuple $c'$ and hence a
  \emph{unique} $w'$
. Thus, if we
  define a map $f^S:\mathcal{D}_S \to \mathbb{R}^+$ by
  \begin{align}
    \label{eqn:fsdef}
    f^S(\phi) \bydef \frac{\sum_{i \in S} c'_i \Lambda^*_i \left(\phi_i\right)}{w'},
  \end{align}
    then (\ref{eqn:j*lb}) is just
  \begin{align} 
    J_* \geq \min_S \inf_{\phi_S \in \mathcal{D}_S} f^S(\phi_S). \label{eqn:j*lb2}
  \end{align}
  
  The next lemma contains the key result needed to prove Proposition \ref{prop:mqopt}:
  \begin{lemma}
    \label{lem:fsmin}
    Let $S \subseteq \{1,\ldots,N\}$ be such that $\mathcal{D}_S \neq \emptyset$. Then,
    \begin{enumerate}
    \item $f^S$ attains its infimum over $\mathcal{D}_S$ at a point
      $\hat{\phi}_S \in \mathcal{D}_S$.
    \item For every $\{c_i'\}_{i \in S}$ with $c_i' \geq
      0$, $\sum_{i \in S} c_i' = 1$, we have
      \[f^S(\hat{\phi}_S) \geq \frac{\sum_{i \in S} c_i' \Lambda^*_i(\hat{\phi}_i)}{\max_{i \in S} (\lambda_i - c_i' \hat{\phi}_i)}. \]
    \end{enumerate}
  \end{lemma}

  \begin{proof}
    Without loss of generality, we will assume $S =
    \{1,\ldots,N\}$. Denote $\mu_i \bydef \mathbb{E}[R_i]$. $\lambda$
    is a stabilizable vector of arrival rates, so $\lambda \in
    \mathcal{C}(\mu_1, \ldots, \mu_N)$ (here $\mu_i$ is overloaded to
    denote the $N$-tuple with the $i$-th coordinate being $\mu_i$ and
    the remaining coordinates being 0). Hence, there exists $\delta >
    0$ such that $\sum_{i=1}^N \frac{\lambda_i}{\mu_i} = 1-\delta$.

    For any $\phi \in \mathcal{D}_S$, we have $\sum_{i=1}^N
    \frac{\lambda_i}{\phi_i} > 1$ by definition. Thus, $\sum_{i=1}^N
    \frac{\lambda_i}{\phi_i} >
    \left(\frac{1}{1-\delta}\right)\sum_{i=1}^N
    \frac{\lambda_i}{\mu_i}$, so $\phi_j < (1-\delta)\mu_j$ for at
    least one $j$. It follows from the properties of the Cram\'{e}r
    rate function for finite alphabets \cite{dembo} that for each $i$,
    $\Lambda^*_i(\cdot)$ is strictly decreasing on
    $[R_{\min,i},\mu_i]$, with $\Lambda^*_i(\mu_i) = 0$. Denote by
    $\gamma$ the positive number $\min_i
    \Lambda^*_i((1-\delta)\mu_i)$. Fix $\epsilon > 0$ small enough. If
    additionally (for $\phi \in \mathcal{D}_S$) $w' <
    \epsilon$, then
    \begin{align*} 
      f^S(\phi) &= \frac{\sum_i c_i' \Lambda^*_i(\phi_i)}{w'} \geq \frac{c_j' \Lambda^*_j(\phi_j)}{w'} = \left(\frac{\lambda_j - w'}{\phi_j w'}\right)\Lambda^*_j(\phi_j) > \left(\frac{\lambda_j - \epsilon}{\mu_j \epsilon} \right) \gamma. 
      \end{align*}
      This means that for every $B > 0$, there exists $\epsilon_B > 0$
      such that $\{\phi \in \mathcal{D}_S: w' < \epsilon_B \}
      \subseteq \{\phi \in \mathcal{D}_S: f^S(\phi) > B\}$. Thus, 
      \[\inf_{\phi \in \mathcal{D}_S} f^S(\phi) = \inf_{\substack{\phi \in \mathcal{D}_S:\\w'\geq \epsilon_B}} f^S(\phi). \]
      Observe that $\{\phi \in \mathcal{D}_S: w' \geq \epsilon_B \}$
      is a compact set, and that the lower-semicontinuity of
      $\Lambda^*_i(\cdot)$ \cite{dembo} forces $f^S$ to be
      lower-semicontinous on this compact set. It follows that $f^S$
      achieves its infimum on this set and thus on
      $\mathcal{D}_S$. This proves the first part of the lemma. 

      Turning to the second part, let $\hat{\phi}_S \in
      \mathcal{D}_S$ infimize $f^S(\cdot)$ over $\mathcal{D}_S$, with
      $R_{\min,i} \leq \hat{\phi}_i \leq \mu_i$ $\forall i \in
      S$. Fix any $i \in S$. Since $\hat{\phi}_S$ is a minimizer,
      increasing $\phi_i = \hat{\phi}_i$ by a small amount (keeping
      the other coordinates unchanged and $\phi_S$ within
      $\mathcal{D}_S$) cannot decrease $f^S(\phi_S)$, i.e.,
      $\left. \frac{\partial}{\partial \phi_i} f^S(\phi_S)
      \right|_{\hat{\phi}_S} \geq 0$. From the definition of $f^S$
      (\ref{eqn:fsdef}), we can write $\frac{\partial}{\partial
        \phi_i} f^S(\phi_S) = \frac{\partial}{\partial \phi_i}
      \frac{N}{D}$, where $N \equiv N(\phi_S) = \sum_{i \in S} c_i'
      \Lambda^*_i(\phi_i)$, and $D \equiv D(\phi_S) = w' \equiv
      w'(\phi_S)$. Thus,
      \begin{align}
        \label{eqn:partialfS}
        0 &\leq \left. \frac{\partial}{\partial \phi_i} f^S(\phi_S)
      \right|_{\hat{\phi}_S} = \frac{1}{D^2(\hat{\phi}_S)} \left. \left( D(\hat{\phi}_S) \frac{\partial}{\partial \phi_i}N({\phi}_S)- N(\hat{\phi}_S) \frac{\partial}{\partial \phi_i}D({\phi}_S) \right)\right|_{\hat{\phi}_S}.
      \end{align}
      Define, for each $i$, $\eta_i' \bydef
      \frac{\Lambda^*_i(\phi_i)}{\phi_i}$ (and $\hat{\eta}_i' \bydef
      \frac{\Lambda^*_i(\hat{\phi}_i)}{\hat{\phi}_i}$). Noticing
      that $\frac{\partial}{\partial \phi_i}D({\phi}_S) =
      \frac{\partial}{\partial \phi_i}(\lambda_j - c_j' \phi_j) = -
      \frac{\partial}{\partial \phi_i} (c_j' \phi_j)$ for all $j \in
      S$, we can write
      \begin{align*}
        \frac{\partial}{\partial \phi_i}N({\phi}_S) &= \frac{\partial}{\partial \phi_i} \sum_{j \in S} c_j' \phi_j \eta_j' = -\frac{\partial D({\phi}_S) }{\partial \phi_i} \cdot \sum_{j \in S} \eta_j' + c_i' \phi_i \cdot \frac{\partial \eta_i'}{\partial \phi_i}.
      \end{align*}
      Along with (\ref{eqn:partialfS}), this implies (evaluated at $\phi_S = \hat{\phi}_S$)
      \begin{align}
        0 &\leq -D(\phi_S) \cdot \frac{\partial D({\phi}_S) }{\partial \phi_i} \cdot \sum_{j \in S} \eta_j' + D(\phi_S) \cdot c_i' \phi_i \cdot \frac{\partial \eta_i'}{\partial \phi_i} - N({\phi}_S) \frac{\partial D({\phi}_S)}{\partial \phi_i} \nonumber \\
        &= -\frac{\partial D({\phi}_S)}{\partial \phi_i} \left[D(\phi_S) \cdot \sum_{j \in S} \eta_j' + N(\phi_S)\right] + D(\phi_S) \cdot c_i' \phi_i \cdot \frac{\partial \eta_i'}{\partial \phi_i} \nonumber \\
        &= -\frac{\partial D({\phi}_S)}{\partial \phi_i} \left[D(\phi_S) \cdot \sum_{j \in S} \eta_j' + N(\phi_S)\right] + D(\phi_S) \cdot c_i' \phi_i \cdot \frac{\phi_i \frac{\partial \Lambda^*_i(\phi_i)}{\partial \phi_i} - \Lambda^*_i(\phi_i)}{\phi^2_i} \nonumber \\
        &= -\frac{\partial D({\phi}_S)}{\partial \phi_i} \left[D(\phi_S) \cdot \sum_{j \in S} \eta_j' + N(\phi_S)\right] - D(\phi_S) c_i' \eta_i' + D(\phi_S) \cdot c_i' \cdot \underbrace{\frac{\partial \Lambda^*_i(\phi_i)}{\partial \phi_i}}_{ \leq 0} \nonumber \\
        &\leq - \underbrace{\frac{\partial D({\phi}_S)}{\partial \phi_i}}_{\leq 0} \left[D(\phi_S) \cdot \sum_{j \in S} \eta_j' + N(\phi_S)\right] - D(\phi_S) c_i' \eta_i' \nonumber \\
        \Rightarrow \; \frac{N}{D} &\geq -\frac{c_i' \eta_i'}{\left(\frac{\partial D}{\partial \phi_i}\right)} - \sum_{j \in S} \eta_j'. \label{eqn:ndlb1}
      \end{align}
      Since $D = \lambda_j - c_j' \phi_j$ and $\sum_{j \in S} c_j' =
      1$, we have 
      \[ \sum_{j \in S} \frac{\lambda_j - D}{\phi_j} = 1 \; \Rightarrow \; D = \frac{\sum_{j \in S} \frac{\lambda_j}{\phi_j} - 1}{\sum_{j \in S} \frac{1}{\phi_j}}.\]
      Using this, some calculus yields 
      \begin{align}
        -\frac{c_i'}{\left(\frac{\partial D}{\partial \phi_i}\right)} &= \phi_i \cdot \sum_{j \in S} \frac{1}{\phi_j} \nonumber \\
        \Rightarrow \; \frac{N}{D} &\geq \eta_i'\phi_i \cdot \sum_{j \in S} \frac{1}{\phi_j}  - \sum_{j \in S} \eta_j' \quad \mbox{(by (\ref{eqn:ndlb1}))} \nonumber \\
        \Rightarrow \; f^S(\hat{\phi}_S) = \frac{N}{D} &\geq \left( \max_{i \in S} \hat{\eta}_i'\hat{\phi}_i \right)\cdot \sum_{j \in S} \frac{1}{\hat{\phi}_j}  - \sum_{j \in S} \hat{\eta}_j'. \label{eqn:ndlb2}
      \end{align}
      Now consider any tuple $\{d_i'\}_{i \in S}$ with $d_i' \geq 0$
      and $\sum_{i \in S} d_i'= 1$. Let $\hat{c}'$ be the (unique)
      tuple corresponding to $\hat{\phi}_S$ such that $0 < \lambda_i -
      \hat{c}_i' \hat{\phi}_i = \lambda_j - \hat{c}_j' \hat{\phi}_j$
      $\forall i,j \in S$. Let $\delta_i' \bydef d_i' - \hat{c}_i'$
      for all $i \in S$, so that $\sum_{i \in S} \delta_i' = 0$, and
      for $t \in [0,1]$, define
      \[ g(t) \bydef \frac{\sum_{i \in S}(\hat{c}_i' + t\delta_i')
        \Lambda^*_i(\hat{\phi}_i)}{\max_{i \in S} (\lambda_i -
        (\hat{c}_i' + t\delta_i')\hat{\phi}_i)}, \] so that $g(0) =
      f^S(\hat{\phi}_S)$. To prove the second part of the lemma, we
      proceed to show that $g(0) \geq g(1)$. First, note that since
      (for $t = 0$) $\lambda_i - \hat{c}_i' \hat{\phi}_i$ is equal
      for all $i \in S$, we can assume without loss of generality that
      $1 \in S$ and that the denominator in the definition of $g(t)$
      above is equal to $\lambda_1 - \hat{c}_1' \hat{\phi}_1 - t
      \delta_1' \hat{\phi}_1 = D(\hat{\phi}_S) - t \delta_1'
      \hat{\phi}_1$, with $\delta_1' \hat{\phi}_1 \leq \delta_i'
      \hat{\phi}_i$ for each $i \in S$. This makes $g(\cdot)$ a
      quotient of affine functions on $[0,1]$, and thus monotone. It
      just remains to show that $g'(t) \leq 0$ for all $t$. 

      Consider
      \begin{align*}
        \frac{d}{dt} g(t) = \frac{d}{dt} \left(\frac{N + t \sum_{i \in S} \delta_i' \Lambda^*(\hat{\phi}_i)}{D - t\delta_1' \hat{\phi}_1}\right) &\leq 0 \\
        \Leftrightarrow \quad D \cdot \sum_{i \in S} \delta_i' \Lambda^*(\hat{\phi}_i) + N \cdot \delta_1' \hat{\phi}_1 &\leq 0 \\
        \Leftrightarrow \quad \frac{\sum_{i \in S} \delta_i' \Lambda^*(\hat{\phi}_i)}{\underbrace{-\delta_1' \hat{\phi}_1}_{> 0}} \leq \frac{N}{D} \\
        \Leftrightarrow \quad \sum_{i \in S} \left( \frac{\delta_i' \hat{\phi}_i}{-\delta_1' \hat{\phi}_1} \right) \hat{\eta}_i' \leq \frac{N}{D}.
      \end{align*}
      By (\ref{eqn:ndlb2}), we will be done if we can show that 
      \[ \left( \max_{j \in S} \hat{\eta}_j'\hat{\phi}_j \right)\cdot \sum_{j \in S} \frac{1}{\hat{\phi}_j}  - \sum_{j \in S} \hat{\eta}_j' \geq \sum_{j \in S} \left( \frac{\delta_j' \hat{\phi}_j}{-\delta_1' \hat{\phi}_1} \right) \hat{\eta}_j'.\]
      But notice that
      \begin{align*}
        \sum_{j \in S} \left( \frac{\delta_j' \hat{\phi}_j}{-\delta_1' \hat{\phi}_1} \right) \hat{\eta}_j' &+ \sum_{j \in S} \hat{\eta}_j' = \sum_{j \in S} \hat{\eta}_j' \left[ 1 + \frac{\delta_j' \hat{\phi}_j}{-\delta_1' \hat{\phi}_1}\right] \\
        &\leq \left( \max_{j \in S} \hat{\eta}_j' \hat{\phi}_j \right) \sum_{j \in S} \frac{1}{\hat{\phi}_j} \left[ 1 + \frac{\delta_j' \hat{\phi}_j}{-\delta_1' \hat{\phi}_1}\right] \\
        &= \left( \max_{j \in S} \hat{\eta}_j' \hat{\phi}_j \right) \sum_{j \in S} \frac{1}{\hat{\phi}_j} + \left( \max_{j \in S} \hat{\eta}_j' \hat{\phi}_j \right) \underbrace{\sum_{j \in S} \frac{\delta_j'}{-\delta_1' \hat{\phi}_1}}_{= 0} \\
        &= \left( \max_{j \in S} \hat{\eta}_j' \hat{\phi}_j \right) \sum_{j \in S} \frac{1}{\hat{\phi}_j}.
      \end{align*}
      This completes the proof of the lemma.
  \end{proof}

  Using this lemma, we can finish the proof of the proposition. Let $S$ be
  a subset of channels that achieves the minimum in (\ref{eqn:j*lb2});
  according to the lemma there exists $\hat{\phi}_S$ that infimizes
  $f^S$ over $\mathcal{D}_S$. Extend $\hat{\phi}_S \in
  \mathbb{R}^{|S|}$ to an $N$-tuple $\hat{\phi}' \in \mathbb{R}^{N}$
  by setting coordinates $i \notin S$ to their respective mean channel
  rates $\mathbb{E}[R_i]$. This means that $\Lambda^*_i(\hat{\phi}_i)
  = 0$ for $i \notin S$, so for any $N$-tuple $e'$ on the simplex,
  because $\sum_{i \in S} e'_i \leq 1$, the lemma gives
  \begin{align*}
    &\frac{\sum_{i \in S} e_i' \Lambda^*_i(\hat{\phi}_i)}{\max_{1 \leq i \leq N} (\lambda_i - e_i' \hat{\phi}_i)} \leq \frac{\sum_{i \in S} e_i' \Lambda^*_i(\hat{\phi}_i)}{\max_{i \in S} (\lambda_i - e_i' \hat{\phi}_i)} \leq f^S(\hat{\phi}_S) \leq J_*,
  \end{align*}
  completing the proof.

\section{Proof of Theorem \ref{thm:mexpopt}}
\label{app:mexpopt}
  Consider an LFSP, specifically the component functions $(\di q, \di
  c, \di g )$, over time $[0,S]$ under the Max-Exp scheduling
  algorithm. Fix a regular point $s \in [0,S]$. Let 
  \[ \mathcal{O}^* \bydef \arg\max_{\alpha \in \mathcal{O}}
  \Phi_\alpha(\di q(s)) \subseteq \mathcal{O} \] be the subcollection of
  ``active'' observable subsets at time $s$, i.e., the subsets picked
  by Max-Exp at $s$. The regularity of point $s$ and the dynamics of
  the Max-Exp rule (Lemma \ref{lem:lfspdynamics}) implies that the
  derivatives $\left. \frac{d}{du} \Psi_\alpha(u) \right|_{u = s} $
  across all $\alpha \in \mathcal{O}^*$, and $\left. \frac{d}{du}
    \Psi(u) \right|_{u = s}$, are equal to $w'$, say. For each such
  $\alpha$,
  \begin{align}
    w' &= \sum_{i \in \alpha} e^{\di q_i(s) + b_i}(\underbrace{\lambda_i(s)}_{= \lambda_i} - v_i(s)) \nonumber \\
    &= \ip{e^{\di q(s) + b}}{\lambda}_\alpha - \ip{e^{\di q(s) + b}}{v(s)}_\alpha \nonumber \\    
    &= \ip{e^{\di q(s) + b}}{\lambda}_\alpha - \di \dot{c}_\alpha(s) \ip{e^{\di q(s) + b}}{\frac{v(s)}{\di \dot{c}_\alpha(s)}}_\alpha \nonumber \\   
    &= \ip{e^{\di q(s) + b}}{\lambda}_\alpha - \di \dot{c}_\alpha(s) \left[ \max_{\eta_\alpha \in V_{\phi_\alpha(s)}} \ip{e^{\di q(s) + b}}{\eta_\alpha}_\alpha \right]. \label{eqn:w'}
  \end{align}
  For notational convenience, let us denote, for each $\alpha$,
  \[   \rho_\alpha \bydef \ip{e^{\di q(s) + b}}{\lambda}_\alpha, \]
  \[    \xi_\alpha \equiv \xi_\alpha\left(\phi_\alpha(s)\right) \bydef \max_{\eta_\alpha \in V_{\phi_\alpha(s)}} \ip{e^{\di q(s) + b}}{\eta_\alpha}_\alpha = \sum_{r \in \mathcal{R}_\alpha} \phi_{\alpha r} (s) \left[\max_{i \in \alpha} \mu^\alpha_{ri} \cdot e^{\di q_i(s)+b_i}  \right]. \]
  With this, (\ref{eqn:w'}) becomes
  \[ w' \equiv w'(\phi_\alpha(s)) = \rho_\alpha - \di
  \dot{c}_\alpha(s) \cdot \xi_\alpha \left(\phi_\alpha(s) \right). \]
  For fixed $\di q(s) = q$, the map $\xi_\alpha: \Pi_\alpha \to
  \mathbb{R}^+$ is linear and hence continuous. Thus, $\xi_\alpha$
  induces a good rate function $\tilde{\Lambda}^*_\alpha$ on
  $\mathbb{R}^+$ \cite{dembo}, given by
  \[\tilde{\Lambda}^*_\alpha(\nu'_\alpha) \bydef \inf \left\{ \Lambda^*_\alpha(\phi_\alpha): \phi_\alpha \in \Pi_\alpha, \; \xi_\alpha(\phi_\alpha) = \nu'_\alpha \right\}.\]
  We have, with $\mathcal{O}^* \subseteq \mathcal{O}^*$ fixed,
  \begin{align}
    &\frac{\sum_{\alpha \in \mathcal{O}^*} \di \dot{c}_\alpha(s) \Lambda^*_\alpha(\phi_\alpha(s))}{\dot{\Psi}(s)} = \frac{\sum_{\alpha \in \mathcal{O}^*} \di \dot{c}_\alpha(s) \Lambda^*_\alpha(\phi_\alpha(s))}{\rho_\alpha - \di \dot{c}_\alpha(s) \cdot \xi_\alpha(\phi_\alpha(s))} \nonumber \\
    &\geq \inf \left\{ \left.\frac{\sum_{\alpha \in \mathcal{O}^*} c'_\alpha \tilde{\Lambda}^*_\alpha(\nu'_\alpha)}{w'}\right| w' > 0, \nu'_\alpha \geq 0, c'_\alpha \geq 0, \sum_{\alpha \in \mathcal{O}^*} c'_\alpha = 1, \rho_\alpha - c'_\alpha \nu'_\alpha = w'  \; \; \forall \alpha \in \mathcal{O}^*  \right\}. \label{eqn:subsetreduction}
  \end{align}
  This exactly corresponds to infimizing the function $f^S$, given in
  (\ref{eqn:fsdef}), over the corresponding domain $\mathcal{D}_S$ for
  the case of singleton observable subsets/individual channels. The
  correspondence becomes clear when, keeping $\di q$ fixed, we
  \emph{identify each observable subset $\alpha$ with a hypothetical
    queue} having an arrival rate of $\rho_\alpha$ and a ``twisted''
  service rate of $\nu'_\alpha$. Under this correspondence, and due to
  the fact that $\tilde{\Lambda}^*$ is a (good) rate function, we can
  employ the same arguments as those in the proof of Lemma
  \ref{lem:fsmin} to get that
  \begin{enumerate}
  \item There exist $\hat{\nu}'_\alpha \geq 0$, $\alpha \in
    \mathcal{O}^*$, determining \emph{unique} $\hat{w}' > 0$ and
    $\hat{c}'_\alpha \geq 0$ feasible for (\ref{eqn:subsetreduction}),
    such that the infimum (\ref{eqn:subsetreduction}) is attained at
    $(\nu'_\alpha)_{\alpha \in \mathcal{O}^*}$.
  \item For every $(d'_\alpha)_{\alpha \in \mathcal{O}^*} \geq 0$
    with $\sum_{\alpha \in \mathcal{O}^*} d'_\alpha = 1$, we have
    \begin{align}
      \label{eqn:nuopt}
      \frac{\sum_{\alpha \in \mathcal{O}^*} \hat{c}'_\alpha \tilde{\Lambda}^*_\alpha(\hat{\nu}'_\alpha)}{\hat{w}'} &\geq \frac{\sum_{\alpha \in \mathcal{O}^*} d'_\alpha \tilde{\Lambda}^*_\alpha(\hat{\nu}'_\alpha)}{\max_{\alpha \in \mathcal{O}^*} (\rho_\alpha - d'_\alpha \hat{\nu}'_\alpha)}.
    \end{align}
  \end{enumerate}
  For each of the optimizing $\hat{\nu}'_\alpha$ above, by the
  lower-semicontinuity of $\Lambda^*_\alpha$, we can find
  $\hat{\phi}'_\alpha \in \Pi_\alpha$ such that
  $\xi_\alpha(\hat{\phi}'_\alpha) = \hat{\nu}'_\alpha$ and
  $\tilde{\Lambda}^*_\alpha(\hat{\nu}'_\alpha) =
  \Lambda^*_\alpha(\hat{\phi}'_\alpha)$. Consider an arbitrary vector
  $(d'_\alpha)_{\alpha \in \mathcal{O}^*} \geq 0$ with $\sum_{\alpha
    \in \mathcal{O}^*} d'_\alpha = 1$. Returning to our original LFSP
  $(\di q, \di c, \di g)$, from (\ref{eqn:subsetreduction}),
  (\ref{eqn:nuopt}) and the previous remark, we can write
  \begin{align}
    \frac{\sum_{\alpha \in \mathcal{O}^*} \di \dot{c}_\alpha(s) \Lambda^*_\alpha(\phi_\alpha(s))}{\dot{\Psi}(\di q(s))} &\geq \frac{\sum_{\alpha \in \mathcal{O}^*} d'_\alpha \Lambda^*_\alpha(\hat{\phi}'_\alpha)}{\max_{\alpha \in \mathcal{O}^*} (\rho_\alpha - d'_\alpha \cdot \xi_\alpha(\hat{\phi}'_\alpha))}. \label{eqn:aftersingleton}
  \end{align}
  Considering any $\alpha \in \mathcal{O}^*$, we have
  \begin{align}
    \rho_\alpha - d'_\alpha \cdot \xi_\alpha(\hat{\phi}'_\alpha) &= \ip{e^{\di q(s) + b}}{\lambda}_\alpha - d'_\alpha \cdot \max_{\eta_\alpha \in V_{\hat{\phi}'_\alpha}} \ip{e^{\di q(s) + b}}{\eta_\alpha}_\alpha \nonumber \\
    &= \ip{e^{\di q(s) + b}}{\lambda}_\alpha - \max_{\upsilon_\alpha \in d'_\alpha V_{\hat{\phi}'_\alpha}} \ip{e^{\di q(s) + b}}{\upsilon_\alpha}_\alpha \nonumber \\
    &= \min_{\upsilon_\alpha \in d'_\alpha  V_{\hat{\phi}'_\alpha}} \sum_{i \in \alpha} e^{\di q_i(s) + b_i} [\lambda_i - \upsilon_{\alpha i}]. \label{eqn:rhsdenom}
  \end{align}
  Thanks to the key Lemma 12.2 in \cite{stol:ldexp}, we have that
  there exist
    \[ l_\alpha > 0, \quad \di q^*_{\alpha i} \in [-\infty,\infty), i \in \alpha, \quad \mbox{and} \]
    \[\upsilon^*_\alpha \in \arg\max_{\upsilon_\alpha \in d'_\alpha V_{\hat{\phi}'_\alpha}} \ip{e^{\di q^*_\alpha + b}}{\upsilon_\alpha}_\alpha \] 
    such that
    \begin{align*}
      \forall i \in \alpha \quad \lambda_i - \upsilon^*_{\alpha i} &= l_\alpha, \quad \mbox{if} \; e^{\di q^*_{\alpha i}} > 0, \\
      \lambda_i - \upsilon^*_{\alpha i} &\leq l_\alpha, \quad \mbox{if} \; e^{\di q^*_{\alpha i}} = 0, \quad \mbox{and} \\
      \min_{\upsilon_\alpha \in d'_\alpha  V_{\hat{\phi}'_\alpha}} \sum_{i \in \alpha} e^{\di q_i(s) + b_i} [\lambda_i - \upsilon_{\alpha i}] &\leq \Psi_\alpha(\di q(s)) \; l_\alpha \leq  \Psi(\di q(s)) \; l_\alpha \\
      \Rightarrow \quad \max_{\alpha \in \mathcal{O}^*} \min_{\upsilon_\alpha \in d'_\alpha  V_{\hat{\phi}'_\alpha}} \sum_{i \in \alpha} e^{\di q_i(s) + b_i} [\lambda_i - \upsilon_{\alpha i}] & \leq  \Psi(\di q(s)) \cdot \max_{\alpha \in \mathcal{O}^*} l_\alpha. \\
    \end{align*}
    Using this with (\ref{eqn:aftersingleton}) and (\ref{eqn:rhsdenom}) yields
    \begin{align}
      &\frac{\sum_{\alpha \in \mathcal{O}^*} \di \dot{c}_\alpha(s) \Lambda^*_\alpha(\hat{\phi}'_\alpha(s))}{\dot{\Psi}(\di q(s))} \geq \frac{\sum_{\alpha \in \mathcal{O}^*} d'_\alpha \Lambda^*_\alpha(\hat{\phi}'_\alpha)}{\Psi(\di q(s)) \cdot \max_{\alpha \in \mathcal{O}^*} l_\alpha} \nonumber \\
      \Rightarrow \quad &\frac{\sum_{\alpha \in \mathcal{O}^*} \di \dot{c}_\alpha(s) \Lambda^*_\alpha(\hat{\phi}'_\alpha(s))}{\left[ \frac{\dot{\Psi}(\di q(s))}{\Psi(\di q(s))} \right]} \geq \frac{\sum_{\alpha \in \mathcal{O}^*} d'_\alpha \Lambda^*_\alpha(\hat{\phi}'_\alpha)}{\max_{\alpha \in \mathcal{O}^*} l_\alpha} \nonumber \\
      \Rightarrow \quad &\frac{\sum_{\alpha \in \mathcal{O}^*} \di \dot{c}_\alpha(s) \Lambda^*_\alpha(\hat{\phi}'_\alpha(s))}{\dot{\Phi}(\di q(s))} \geq \frac{\sum_{\alpha \in \mathcal{O}^*} d'_\alpha \Lambda^*_\alpha(\hat{\phi}'_\alpha)}{\max_{\alpha \in \mathcal{O}^*} l_\alpha} \nonumber \\
      &\geq \frac{\sum_{\alpha \in \mathcal{O}^*} d'_\alpha \Lambda^*_\alpha(\hat{\phi}'_\alpha)}{\max_{\alpha \in \mathcal{O}^*} \max_{i \in \alpha} (\lambda_i - \upsilon^*_{\alpha i})} \nonumber \\
      &\geq \frac{\sum_{\alpha \in \mathcal{O}^*} d'_\alpha \Lambda^*_\alpha(\hat{\phi}'_\alpha)}{\max_{\alpha \in \mathcal{O}^*} \max_{\upsilon_\alpha \in d'_\alpha V_{\hat{\phi}'_\alpha}} \max_{i \in \alpha} (\lambda_i - \upsilon_{\alpha i})} \nonumber \\
      &\geq \frac{\sum_{\alpha \in \mathcal{O}^*} d'_\alpha \Lambda^*_\alpha(\hat{\phi}'_\alpha)}{\max_{\alpha \in \mathcal{O}^*, v_\alpha \in V_{\hat{\phi}'_\alpha}} \max_{i \in \alpha} (\lambda_i - d'_\alpha v_{\alpha i})}. \label{eqn:ubconnection1}
    \end{align}
    The above relation holds for any $(d'_\alpha)_{\alpha \in
      \mathcal{O}^*} \geq 0$ with $\sum_{\alpha \in \mathcal{O}^*}
    d'_\alpha = 1$. Let $(c'_\alpha)_{\alpha \in \mathcal{O}} \geq 0$
    be such that $\sum_{\alpha \in \mathcal{O}} c'_\alpha = 1$.  For
    each $\alpha \in \mathcal{O} \setminus \mathcal{O}^*$, define
    $\hat{\phi}'_\alpha$ to be the natural probability distribution of
    sub-states in $\alpha$, so that
    $\Lambda^*_\alpha(\hat{\phi}'_\alpha) = 0$ for such $\alpha$. We
    can write,
    \begin{align}
      &\frac{\sum_{\alpha \in \mathcal{O}} c'_\alpha \Lambda^*_\alpha(\hat{\phi}'_\alpha)}{\max_{\alpha \in \mathcal{O}, v_\alpha \in V_{\hat{\phi}'_\alpha}} \max_{i \in \alpha} (\lambda_i - c'_\alpha v_{\alpha i})} = \frac{\sum_{\alpha \in \mathcal{O}^*} c'_\alpha \Lambda^*_\alpha(\hat{\phi}'_\alpha)}{\max_{\alpha \in \mathcal{O}, v_\alpha \in V_{\hat{\phi}'_\alpha}} \max_{i \in \alpha} (\lambda_i - c'_\alpha v_{\alpha i})} \nonumber \\
      &\leq \frac{\sum_{\alpha \in \mathcal{O}^*} c'_\alpha \Lambda^*_\alpha(\hat{\phi}'_\alpha)}{\max_{\alpha \in \mathcal{O}^*, v_\alpha \in V_{\hat{\phi}'_\alpha}} \max_{i \in \alpha} (\lambda_i - c'_\alpha v_{\alpha i})} \nonumber \\
      &\leq \frac{\sum_{\alpha \in \mathcal{O}^*} \tilde{c}'_\alpha \Lambda^*_\alpha(\hat{\phi}'_\alpha)}{\max_{\alpha \in \mathcal{O}^*, v_\alpha \in V_{\hat{\phi}'_\alpha}} \max_{i \in \alpha} (\lambda_i - \tilde{c}'_\alpha v_{\alpha i})}, \quad \mbox{where} \; \tilde{c}' \bydef \frac{c'}{\sum_{\alpha \in \mathcal{O}} c'_\alpha}. \label{eqn:ubconnection2}
    \end{align}
    Putting (\ref{eqn:ubconnection1}) and (\ref{eqn:ubconnection2})
    together, we have, for our original LFSP, that
    \begin{align}
      \frac{\sum_{\alpha} \di \dot{c}_\alpha(s) \Lambda^*_\alpha(\hat{\phi}'_\alpha(s))}{\dot{\Phi}(\di q(s))} &\geq \sup_{\substack{\sum_\alpha c_\alpha' = 1\\c_\alpha' \geq 0}} \left[ \frac{\sum_{\alpha} c'_\alpha \Lambda^*_\alpha(\hat{\phi}'_\alpha)}{\max_{\alpha, v_\alpha \in V_{\hat{\phi}'_\alpha}} \max_{i \in \alpha} (\lambda_i - c'_\alpha v_{\alpha i})} \right] \nonumber \\
      &\geq -\liminf_{n \to \infty} \frac{1}{n} \log \mathbb{P}^\pi\left[||q^{(n)}(0)||_\infty \geq 1\right], \label{eqn:lfsptoub}
    \end{align}
    for the stationary measure $\mathbb{P}^\pi$ of any stabilizing
    scheduling policy, by Theorem \ref{thm:ubanysubset}. Infimizing
    (\ref{eqn:lfsptoub}) over all valid LFSPs and using Proposition
    \ref{prop:lfsplb} yields
    \begin{align*}
      J_* &\geq -\liminf_{n \to \infty} \frac{1}{n} \log \mathbb{P}^\pi\left[||q^{(n)}(0)||_\infty \geq 1\right],
    \end{align*}
    which finishes the proof.

\end{document}